\documentclass[11pt]{article}
\usepackage{amsmath, amssymb, amsthm, amsfonts}
\usepackage[utf8]{inputenc}
\usepackage{geometry}
\usepackage{hyperref}
\usepackage{array}
\usepackage{placeins}
\usepackage{tikz-cd}
\usepackage{url}
\newtheorem{theorem}{Theorem}[section]
\newtheorem{definition}[theorem]{Definition}
\newtheorem{example}[theorem]{Example}
\newtheorem{remark}[theorem]{Remark}
\newtheorem{lemma}[theorem]{Lemma}
\newtheorem{corollary}[theorem]{Corollary}

\newtheorem{proposition}[theorem]{Proposition}

\newtheorem{fact}[theorem]{Fact}

\newcommand{\A}{\mathcal{A}}

\newcommand{\Tr}{\operatorname{Tr}}

\newcommand{\Herm}{\mathcal{H}_{n}}

\DeclareMathOperator{\diag}{diag}
\newcommand{\Ad}{\operatorname{Ad}}

\newcommand{\Sn}{\mathcal{S}^{\mathrm n}}
\title{Finite-Precision Algebraic Quantum Field Theory: 
Quantum Parcels, Locality, Modular Structure and  Factors }
\author{Abbas Edalat\thanks{Corresponding author: Abbas Edalat, a.edalat@imperial.ac.uk}\\Imperial College London\\a.edalat@imperial.ac.uk}
\date{}

\begin{document}

\maketitle

\begin{abstract}
We develop Interval Algebraic Quantum Field Theory (IAQFT), a
finite-precision operational formulation of algebraic quantum field
theory (AQFT). Exact states are indispensable mathematical objects of
AQFT, but physical experiments provide only finitely many observations
at finite resolution and therefore determine regions of compatible
states rather than uniquely specified point states. IAQFT represents
this finite observational information by quantum parcels---non-empty
convex weak$^*$ open regions of state space---with observational parcels
determined directly by finitely many finite-precision expectation
constraints. The local and operator-algebraic structure of AQFT is
unchanged; exact states are retained as ideal limiting objects recovered
under appropriate parcel refinement.

We establish reduction to exact states under successive observational
refinement, measurement update and no-signalling, and formulate locality
through compatible parcel nets. In finite dimensions an exact Jacobian
formula for invertible L\"uders updates yields a general
volume-contraction criterion. Spacelike vacuum correlations and strict
CHSH violations, when present, persist throughout sufficiently small
observational parcels. Reeh--Schlieder cyclicity becomes finite-precision local
reachability, while parcel equivalence of vacuum refinements implies
vacuum GNS equivalence and hence does not evade Haag-type obstructions.

The modular structure of AQFT admits a corresponding finite-precision
formulation. Tomita--Takesaki modular automorphisms act naturally on
parcels, measurement update is modularly covariant, and modular
conjugation induces a duality between normal parcels of a von Neumann
algebra and its commutant. State-dependent modular correlations are
continuous under predual-norm variation of faithful normal states,
while shrinking parcels with vanishing KMS defect recover a KMS state
for fixed dynamics. Bisognano--Wichmann yields a parcel formulation of
the Unruh effect, with the exact KMS relation and Unruh temperature
recovered under refinement. The framework is also compatible with
directed lattice approximations.

Finally, limiting projection faces and their canonical conjugation
actions characterize minimality and finiteness and recover the
Murray--von Neumann classification of factors. For the
Type~$\mathrm{III}$ factors characteristic of local AQFT, the absence of
minimal projections and pure normal states is reflected in the absence
of singleton nonzero projection faces; in Type~$\mathrm{II}_1$ factors,
homogeneous parcel counting reconstructs the normalized trace and
projection equivalence. Thus IAQFT provides an operational state description appropriate to
finite experimental information, while retaining exact AQFT states as
ideal limiting objects; within this framework it distinguishes
structures robust at finite precision, those acquiring a finite-precision
operational interpretation, and those recovered only in the exact limit.
\end{abstract}

\section{Introduction}
\label{sec:intro}

Algebraic quantum field theory (AQFT), developed by Haag, Kastler, Araki
and others, provides a mathematically rigorous formulation of relativistic
quantum field theory in terms of nets of local operator algebras
\cite{HaagKastler,Araki,haag2012local}. Together with the Wightman
framework~\cite{Wightman1956,streaterWightman}, it arose in part from the
effort to place quantum field theory on firm mathematical foundations
beyond the formal perturbative expansions and Feynman-diagram methods that
are extraordinarily successful computationally but raise difficult
questions concerning divergences, renormalisation and interacting fields.
AQFT has led to major structural results, including the Reeh--Schlieder
theorem, Tomita--Takesaki modular theory, the Bisognano--Wichmann theorem,
superselection theory and the Murray--von Neumann analysis of local
von Neumann algebras.

These structures are formulated in terms of exact states, exact
expectation values and exact algebraic relations. Such exact objects are
indispensable to the mathematical theory, but an exactly specified point
state is an idealisation from the operational point of view. A physical
experiment supplies only finitely many observations, each with finite
resolution, and therefore determines a region of states compatible with
the observations rather than a unique point of state space. The question
is consequently not whether finite precision should be added as a
correction to an otherwise operationally exact state description, but how
the finite information actually available in experiments should be
represented within the exact mathematical structure of AQFT.

This distinction between exact mathematical description and operational
accessibility has been present since the foundations of quantum theory.
Heisenberg's original formulation of quantum mechanics was guided by the
principle that the theory should be expressed in terms of observable
quantities~\cite{Heisenberg1925}, while von Neumann likewise emphasized
the limited operational significance of complete microscopic
specification, particularly for macroscopic systems~\cite{vonNeumann1955}.

The finite-precision viewpoint was developed for ordinary quantum
mechanics in Interval Quantum Mechanics (IQM)~\cite{Edalat2026}. There a
finite experimental specification is represented by an
\emph{observational parcel}, defined by finitely many finite-precision
expectation constraints on density matrices, while a \emph{quantum
parcel}, or simply a parcel, is a non-empty convex weak open set of
density matrices. Exact density operators remain part of the mathematical
theory and are recovered as ideal limits of successive parcel refinement.
The same viewpoint has also been studied in the context of ergodicity and
thermalisation~\cite{EdalatErgodicity2026}.

The present paper extends this principle to AQFT. The mathematical
setting is substantially different: AQFT begins with the state space
$\mathcal S(\mathcal A)$ of a quasi-local $C^*$-algebra, naturally
equipped with the weak$^*$ topology; where von Neumann algebras are
required, the relevant objects are normal states. We define a
\emph{quantum parcel}, or simply a \emph{parcel}, to be a non-empty
convex weak$^*$ open region of state space. When the region is determined
directly by finitely many finite-precision expectation constraints, it is
called an \emph{observational parcel}; such parcels form a basis for the
relative weak$^*$ topology. The local algebras, their net structure,
locality, covariance and modular structure remain unchanged: IAQFT
modifies the operational representation of state information, not the
underlying operator-algebraic theory.

The first part of the theory develops the basic consequences of this
replacement of operationally specified point states by parcels.
Measurement update acts naturally on finite-precision state information,
and global reduction gives conditions under which successive refinement
recovers a unique exact state. In finite dimensions IAQFT is canonically
equivalent to IQM; observational coordinates provide an operational
volume, and an exact Jacobian formula for invertible L\"uders updates
gives a general volume-contraction criterion. Thus exact states remain
the limiting mathematical objects while parcels carry the information
available at finite experimental resolution.

Locality requires the parcel description to respect the net structure of
AQFT. Compatible parcel nets assign finite-precision state information to
nested spacetime regions consistently under restriction, and a local
reduction theorem gives conditions under which decreasing observational
oscillations determine a unique limiting local state. Nonzero spacelike
vacuum correlations and strict CHSH violations, when present, persist
throughout sufficiently small observational parcels, while nonselective
local measurement preserves spacelike-separated expectation values.
Finite observational uncertainty therefore retains both the
nonclassical correlations and the no-signalling structure of AQFT.

With this basic framework in place, we turn to characteristic structural
results of AQFT. The Reeh--Schlieder theorem states that, for every
non-empty bounded region $\mathcal O$, the vacuum vector is cyclic for
the local algebra. Its operational interpretation has long appeared
counterintuitive~\cite{Halvorson2001,VALENTE2014147}, since mathematical
density can be read too readily as the physical ability to prepare an
arbitrarily specified global target state by acting locally. At finite
precision the natural question is instead whether every parcel around a
target state is locally reachable. We show that this local parcel
reachability is precisely equivalent to Reeh--Schlieder cyclicity.

Haag's theorem raises a complementary question. Under its standard
hypotheses, free and interacting theories cannot be related by the
vacuum-preserving unitary equivalence suggested by the interaction
picture~\cite{haag2012local,EarmanFraser2006}. Since physical comparison
has finite observational content and precision, one may ask whether
passage from exact states to parcels weakens this obstruction. For the
notion of parcel equivalence introduced below, it does not: equivalence
of complete vacuum refinement structures already forces equivalence of
the limiting vacuum GNS representations. Haag-type inequivalence
therefore survives passage to finite-precision state information.

Tomita--Takesaki theory admits a corresponding parcel formulation.
Modular automorphisms act naturally on normal parcels, measurement update
is modularly covariant, and modular conjugation induces a duality between
normal parcels of a von Neumann algebra and its commutant. State-dependent
modular correlations are continuous under predual-norm variation of
faithful normal states, while shrinking parcels with vanishing KMS defect
recover a KMS state for fixed dynamics. The Bisognano--Wichmann theorem
then yields a finite-precision formulation of the Unruh effect: Rindler
parcels determine ranges of thermal correlation values, while refinement
around the Minkowski vacuum recovers the exact KMS relation and Unruh
temperature
\cite{bisognanowichmann76,Unruh1976,CrispinoHiguchiMatsas2008}.

The von Neumann-algebraic structure of local AQFT provides a further
test of the framework. Local algebras are typically Type~$\mathrm{III}$
factors~\cite{Yngvason2005,FewsterRejzner2020}; they have no minimal
projections or pure normal states and admit no normal trace supporting
the familiar finite-dimensional density-matrix description
\cite{Yngvason2005,Witten2018}. Parcels require none of these structures.
Moreover, limiting projection faces and their canonical conjugation
actions recover substantial Murray--von Neumann structure: singleton
faces characterize minimality, invariant states characterize finiteness,
and these criteria determine factor type. In Type~$\mathrm{II}_1$
factors, homogeneous parcel counting reconstructs the normalized trace
and projection equivalence.

The framework is also compatible with directed lattice approximations.
Lattice refinement approximates the continuum observable algebra, whereas
observational refinement increases the precision with which a state is
specified. Under the stated inductive-limit hypotheses these two
processes are compatible: lattice parcels lift to continuum parcels and
parcel reduction commutes with passage to the continuum limit.

The use of open sets and refinement orders has a natural
information-theoretic background. In geometric logic and domain theory,
finite information is represented by open regions or partial
approximations, with increasing information corresponding to refinement
\cite{johnstone1982stone,vic89,abramsky1991domain,scott1970outline,
smyth1983,abramsky_domain_1994}. This viewpoint motivates the use of
weak$^*$ open parcels and their ordering by reverse inclusion. IAQFT does
not require domain-theoretic machinery for its operator-algebraic
results, but adopts the underlying principle that finite information is
naturally represented by regions of possibilities rather than exact
points.

The relation between AQFT and IAQFT may also be understood through an
analogy with numerical computation. Real analysis provides the exact
mathematical description of the continuum, yet practical computation
uses finite representations such as floating-point numbers or interval
approximations. Likewise, AQFT provides the exact mathematical theory,
whereas IAQFT represents experimentally available state information by
finite-precision parcels. Exact states remain indispensable and are
recovered as ideal limits of refinement, just as real numbers are
approached by increasingly accurate finite representations.

IAQFT should be distinguished from perturbative algebraic quantum field
theory (pAQFT)~\cite{Rejzner2016}, which constructs interacting theories
perturbatively using topological algebras of functionals and formal power
series. IAQFT instead works within the standard $C^*$- and von
Neumann-algebraic setting and introduces finite precision at the level of
operational state information.

Related finite-information approaches pursue different aims. Del Santo
and Gisin~\cite{DelSantoGisin2025} investigate whether features usually
regarded as distinctively quantum may arise from finite information and
fundamental indeterminism. IAQFT instead retains the operator-algebraic
and nonclassical structure of quantum field theory, including Bell
violation, while changing its operational state description. In
\cite{Treppiedi2026}, finite operational resolution is imposed through
consistency conditions on representations of additive symmetry and used
to derive quantum-mechanical and quantum-field-theoretic structure.
IAQFT does not derive the underlying formalism from finite resolution;
it retains AQFT and asks how finite observational information about its
states should be represented.

The resulting picture distinguishes several ways in which exact AQFT
structure appears at finite precision. Some properties, such as
spacelike correlations and strict Bell violations, are already robust on
sufficiently small observational parcels; some, such as
Reeh--Schlieder cyclicity, acquire a finite-precision operational
interpretation; structural obstructions such as Haag-type inequivalence
survive passage to parcels; and exact states and thermal relations emerge
under successive refinement. IAQFT thereby places finite observational
information and exact operator-algebraic structure within a single
framework.

The paper is organised as follows. Section~\ref{sec:state-space} recalls
the required state-space background. Sections~\ref{sec:parcels-update}--
\ref{sec:nets-germs} develop parcels, measurement, reduction,
finite-dimensional volume contraction and locality.
Sections~\ref{sec:Bell-correlation}--\ref{sec:haag} treat spacelike
correlations, Bell violation, Reeh--Schlieder and Haag-type obstructions.
Sections~\ref{sec:modular}--\ref{sec:unruh} develop modular, KMS and
Unruh theory, and Section~\ref{sec:lattice} treats lattice approximation.
Section~\ref{sec:minimality-finiteness} develops the factor-theoretic
results. The final sections discuss open problems and conclude with the
conceptual consequences of IAQFT.
\section{State Spaces and the Weak* Topology}
\label{sec:state-space}

We assume familiarity with the basic theory of $C^*$- and von Neumann
algebras, including states, representations and the GNS construction, and
with the Haag--Kastler formulation of AQFT; standard references include
\cite{KadisonRingrose1,KadisonRingrose2,brattelirobinson1,
haag2012local,Araki}. We recall only the state-space and topological facts
needed for the finite-precision construction below. For background on
weak and weak$^*$ topologies, Banach--Alaoglu compactness and state spaces
of $C^*$-algebras, see, for example,
\cite{KadisonRingrose1,brattelirobinson1}.

Let $\A$ be a unital $C^*$-algebra. A \emph{state} on $\A$ is a linear
functional $\omega:\A\to\mathbb C$ satisfying $\omega(A^*A)\geq0$ for all
$A\in\A$ and $\omega(I)=1$. The state space
$\mathcal S(\A)\subseteq\A^*$ is convex and is equipped with the
weak$^*$ topology
\[
\sigma(\A^*,\A),
\]
the coarsest topology making every evaluation
$\omega\mapsto\omega(A)$, $A\in\A$, continuous. Thus a net
$\{\omega_\alpha\}$ converges weak$^*$ to $\omega$ precisely when
\[
\omega_\alpha(A)\longrightarrow\omega(A)
\qquad\text{for every }A\in\A.
\]

For each $A\in\A$, the condition $\omega(A^*A)\geq0$ is weak$^*$
closed, as is the condition $\omega(I)=1$. Hence
$\mathcal S(\A)$ is a weak$^*$ closed subset of the closed unit ball of
$\A^*$ and is therefore weak$^*$ compact by the Banach--Alaoglu theorem.
It is Hausdorff because distinct states differ on some $A\in\A$. 

Since every $A\in\A$ can be written $A=B+iC$ with $B,C$ self-adjoint,
the relative weak$^*$ topology on $\mathcal S(\A)$ has a basis consisting
of finite intersections of expectation constraints
\[
\left\{
\omega\in\mathcal S(\A):
a_i<\omega(A_i)<b_i,\quad i=1,\ldots,m
\right\},
\qquad
A_i=A_i^*,\quad a_i,b_i\in\mathbb R.
\]
Thus finitely many finite-precision expectation measurements naturally
determine weak$^*$ open regions of state space. These will be formalized
as observational parcels in the next section, where general quantum
parcels are defined as non-empty convex weak$^*$ open subsets of
$\mathcal S(\A)$.

An AQFT is specified by a net of local $C^*$-algebras
$\{\A(\mathcal O)\}_{\mathcal O}$ indexed by bounded spacetime regions
and satisfying isotony
\cite{HaagKastler,haag2012local}:
\[
\mathcal O_1\subseteq\mathcal O_2
\quad\Longrightarrow\quad
\A(\mathcal O_1)\subseteq\A(\mathcal O_2).
\]
More formally, these inclusions may be represented by compatible
injective unital $*$-homomorphisms. After identifying the local algebras
with their images, the associated \emph{quasi-local algebra} is
\[
\A
=
\overline{\bigcup_{\mathcal O}\A(\mathcal O)}^{\,\|\cdot\|}.
\]
The quasi-local state space $\mathcal S(\A)$ provides the global state
space for the finite-precision construction introduced in the next
section, while locality remains encoded by the net of local algebras and
the induced restriction maps on their state spaces.
\section{Quantum Parcels and Measurement Update}
\label{sec:parcels-update}
A physical experiment does not determine the expectation values of
observables exactly. Measurements are finite in number and are performed
with apparatus of finite resolution and accuracy, so the resulting
observational information constrains expectation values to lie within
finite intervals rather than assigning them exact real values. Such
information therefore does not single out a unique state: it determines
the collection of all states whose expectation values satisfy the
observed constraints. Since finite intersections of expectation intervals
form a basis for the weak$^*$ topology, these finite-precision state
descriptions are naturally represented by open regions of state space.
Requiring convexity additionally reflects the fact that if two states are
compatible with the available information, then so is any probabilistic
mixture of them. This motivates the central operational object of IAQFT.
\begin{definition}[Observational and quantum parcels]
\label{def:quantum-parcel}
Let $\A$ be a unital $C^*$-algebra.
An \emph{observational parcel} is a non-empty set of the form
\[
O
=
\left\{
\omega\in\mathcal S(\A):
a_i<\omega(A_i)<b_i,\quad i=1,\ldots,m
\right\},
\]
where $A_i=A_i^*\in\A$ and $a_i,b_i\in\mathbb R$ with $a_i<b_i$.
It represents the states compatible with finitely many finite-precision
expectation-value observations.

A \emph{quantum parcel}, or simply a \emph{parcel}, is a non-empty
convex weak$^*$ open subset
\[
O\subseteq\mathcal S(\A).
\]
\end{definition}

Every observational parcel is a parcel, since each evaluation
$\omega\mapsto\omega(A_i)$ is weak$^*$ continuous and affine.
Observational parcels form a basis for the relative weak$^*$ topology on
$\mathcal S(\A)$. Thus observational parcels describe finite
experimental specifications directly, while parcels provide the natural
class of finite-precision state regions used throughout IAQFT.

The three requirements have direct operational meanings: openness represents
finite observational tolerance, convexity expresses closure under
probabilistic mixing, and non-emptiness ensures compatibility with at least
one state. We order parcels by reverse inclusion,
$O\sqsubseteq U$ iff $O\supseteq U$, so that refinement corresponds to
increasing information.

\begin{example}[A finite-precision observational parcel]
Let $A=A^*\in\A$ and suppose an experiment determines its expectation
value only to lie in the interval $(a,b)$. The compatible states form the observational parcel
\[
O(A;a,b)
=
\{\omega\in\mathcal S(\A):a<\omega(A)<b\}.
\]
With two observed quantities $A,B$, the corresponding observational
parcel is
\[
O(A,B;a,b,c,d)
=
\left\{
\omega\in\mathcal S(\A):
a<\omega(A)<b,\quad
c<\omega(B)<d
\right\}.
\]
Thus increasing the number or precision of observational constraints
shrinks the region of compatible states without requiring an exact state
to be specified.
\end{example}

\subsection{Operational Volume and Information}
\label{sec:operational-volume}
An observational parcel may lie in an infinite-dimensional state space
while being specified by only finitely many observations. Its finite
observational content therefore admits a natural Euclidean representation.

Let
$A_1,\ldots,A_m\in\A$ be self-adjoint observables. Discarding linear
dependencies and constant directions proportional to $I$, choose
self-adjoint observables $H_1,\ldots,H_d$ whose expectation values give
independent observational coordinates, and let
\[
V=\operatorname{span}_{\mathbb R}\{H_1,\ldots,H_d\}.
\]
Define
\[
\Phi_V:\mathcal S(\A)\longrightarrow\mathbb R^d,
\qquad
\Phi_V(\omega)
=
\bigl(\omega(H_1),\ldots,\omega(H_d)\bigr).
\]
For an observational parcel, the defining observations naturally
determine such a finite-dimensional observational space $V$. More
generally, for any parcel $O$ and any fixed finite-dimensional
observational space $V$, the set $\Phi_V(O)$ is its
\emph{observational shadow} in the chosen coordinates.

\begin{definition}[Operational volume]
\label{def:operational-volume}
Let $O$ be a parcel and let $V$ be a fixed finite-dimensional
observational space as above. If $\Phi_V(O)$ has full dimension in
$\mathbb R^d$, define the operational volume of $O$ relative to $V$ by
\[
\operatorname{Vol}_V(O)
:=
\operatorname{Vol}_{\mathbb R^d}\bigl(\Phi_V(O)\bigr),
\]
where $\operatorname{Vol}_{\mathbb R^d}$ denotes Lebesgue volume on
$\mathbb R^d$.
\end{definition}

Operational volume is not an intrinsic volume on the generally
infinite-dimensional state space. It measures only the finite
observational information represented in the chosen coordinates $V$.
In particular, it depends on the observational subspace and on the
normalization of its coordinates. Once these are fixed, however, it is
monotone under parcel refinement: if $U\subseteq O$, then
$\Phi_V(U)\subseteq\Phi_V(O)$ and hence
\[
\operatorname{Vol}_V(U)
\leq
\operatorname{Vol}_V(O).
\]

This gives a corresponding relative measure of information.

\begin{definition}[Operational information]
\label{def:operational-information}
Let $U\subseteq O$ be parcels whose observational shadows in the same
fixed observational space $V$ have positive finite $V$-volume. Define the information gained in refining $O$ to $U$
by
\[
I_V(U:O)
:=
\log
\frac{\operatorname{Vol}_V(O)}
     {\operatorname{Vol}_V(U)}.
\]
\end{definition}

Although the individual volumes depend on the choice and normalization
of observational coordinates, the ratio $I_V(U:O)$ is invariant under
any invertible linear change of coordinates on the fixed observational
space $V$.

Thus $I_V(U:O)\geq0$, with strict inequality whenever the observational
volume decreases strictly. More generally, for
$W\subseteq U\subseteq O$,
\[
I_V(W:O)
=
I_V(W:U)+I_V(U:O),
\]
so successive refinement gives additive information gain.

The construction separates the generally infinite-dimensional geometry
of the AQFT state space from the finite-dimensional information supplied
by an experiment. It requires no trace or density-matrix representation
and therefore remains meaningful for observational parcels on
Type~$\mathrm{III}$ state spaces. In finite dimensions, Section~\ref{sec:finite} identifies the
operational volume associated with an informationally complete
observational space with Hilbert--Schmidt volume and obtains quantitative
measurement-contraction results. Whenever such an update strictly
decreases the corresponding volume, the associated relative operational
information gain is strictly positive.

An observational parcel defined by finitely many local observables is
already determined at a finite stage of the AQFT net. We assume, as is standard, that the
index set of bounded regions is directed under inclusion, so that any
finitely many regions $\mathcal O_1,\ldots,\mathcal O_m$ are contained in
some common bounded region $\mathcal O$; if $A_i\in\A(\mathcal O_i)$ for
$i=1,\ldots,m$, isotony then gives $A_i\in\A(\mathcal O)$ for all $i$. An
arbitrary element of the quasi-local algebra need not itself be local,
although it is a norm limit of local observables.

\subsection{Measurement Update}
\label{sec:measurement-update}

Let $\{E_j\}_{j=1}^m\subset\A$ be a POVM, so that $E_j\geq0$ and
$\sum_jE_j=I$, and write $E_j=M_j^*M_j$. Conditional on outcome $j$, a
state with $\omega(E_j)>0$ is updated by
\[
f_j(\omega)(A)
=
\frac{\omega(M_j^*AM_j)}{\omega(E_j)}.
\]
The corresponding nonselective update is
$\mathcal E^*(\omega)(A)=\sum_j\omega(M_j^*AM_j)$. The former is
nonlinear because of normalization, whereas the latter is linear and
always defined. Selective updates describe conditional parcel evolution;
the nonselective update will be used below in connection with
no-signalling.

If $M_j$ is invertible, $f_j$ is a fractional-linear weak$^*$
homeomorphism of its domain; see, for example,
\cite[Sec.~2.3]{brattelirobinson1}. It also maps line segments to line
segments up to reparametrisation.

\begin{proposition}[Stability under invertible L\"uders updates]
\label{prop:fj}
If $M_j$ is invertible, then $f_j(O)$ is a parcel for every parcel $O$.
\end{proposition}

\begin{proof}
The weak$^*$ homeomorphism property gives openness and preserves
non-emptiness. For convexity, let
$\omega=\lambda\omega_1+(1-\lambda)\omega_2$, with
$0\leq\lambda\leq1$. Then
\[
f_j(\omega)
=
\mu f_j(\omega_1)+(1-\mu)f_j(\omega_2),
\qquad
\mu=
\frac{\lambda\omega_1(E_j)}
{\lambda\omega_1(E_j)+(1-\lambda)\omega_2(E_j)}
\in[0,1].
\]
Hence $f_j$ maps line segments into line segments and therefore maps
convex sets to convex sets.
\end{proof}

For rank-deficient $M_j$, the image need not be weak$^*$ open in
$\mathcal S(\A)$ and may instead lie in the face associated with the range
projection of $E_j$; Proposition~\ref{prop:fj} therefore does not apply.

\begin{definition}[Measurement update]
\label{def:update}
For a parcel $O$ and observed outcome $j$, define
\[
O'
=
f_j\bigl(O\cap\{\omega:\omega(E_j)>0\}\bigr).
\]
If the intersection is empty, outcome $j$ is impossible for every state
in $O$.
\end{definition}

For invertible $M_j$, Proposition~\ref{prop:fj} ensures that the updated
set is again a parcel. Moreover, the update respects the information
order: a more informative initial parcel gives a more informative updated
parcel for the same observed outcome.
\begin{proposition}[Monotonicity of measurement update]
\label{prop:update-monotone}
Let $M_j$ be invertible. If $O\sqsubseteq U$, then
\[
f_j(O)\sqsubseteq f_j(U).
\]
Thus selective measurement update preserves the information order.
\end{proposition}

\begin{proof}
Since $O\sqsubseteq U$ means $O\supseteq U$, we have
$f_j(O)\supseteq f_j(U)$. By Proposition~\ref{prop:fj} both images are
parcels, so this is precisely $f_j(O)\sqsubseteq f_j(U)$.
\end{proof}

Thus IAQFT lifts the usual state update to finite-precision regions of
state space. Exact state dynamics is retained and recovered when parcel
refinement approaches a point state.
\section{Reduction Theory}\label{sec:reduction}

A central question in IAQFT is whether increasingly precise parcel
descriptions recover the exact states of AQFT. The first result treats a
prescribed limiting state; the second shows that shrinking observational
variation alone determines a unique state.

\begin{theorem}[Reduction to a prescribed state]
\label{thm:prescribedstate}
Let $\A$ be a unital $C^*$-algebra, $s\in\mathcal S(\A)$, and
$\{O_\alpha\}_{\alpha\in I}$ a decreasing net of parcels containing $s$.
Let $\A_0\subseteq\A$ be a norm-dense $*$-subalgebra. If, for every
$A\in\A_0$,
\[
\sup_{\omega\in O_\alpha}|\omega(A)-s(A)|\longrightarrow0,
\]
then $\bigcap_\alpha O_\alpha=\{s\}$.
\end{theorem}

\begin{proof}
Clearly $s\in\bigcap_\alpha O_\alpha$. If
$\omega\in\bigcap_\alpha O_\alpha$, then for every $A\in\A_0$ and every
$\alpha$,
\[
|\omega(A)-s(A)|
\leq
\sup_{\eta\in O_\alpha}|\eta(A)-s(A)|.
\]
Hence $\omega(A)=s(A)$ on $\A_0$. Norm density of $\A_0$ and norm
continuity of states imply $\omega=s$ on $\A$.
\end{proof}

The preceding result assumes the limiting state in advance. The parcels
themselves determine the limit under an intrinsic shrinking condition.

\begin{theorem}[Reduction by parcel refinement]
\label{thm:reduction}
Let $\{O_\alpha\}_{\alpha\in I}$ be a decreasing net of parcels in
$\mathcal S(\A)$ and set $C_\alpha=\overline{O_\alpha}$. Let
$\A_0\subseteq\A$ be a norm-dense $*$-subalgebra and suppose that, for
every $A\in\A_0$,
\[
\operatorname{osc}_{O_\alpha}(A)
:=
\sup_{\omega,\eta\in O_\alpha}|\omega(A)-\eta(A)|
\longrightarrow0.
\]
Then $\bigcap_\alpha C_\alpha$ consists of exactly one state.
\end{theorem}

\begin{proof}
Each $C_\alpha$ is non-empty and weak$^*$ compact, and the decreasing
family $\{C_\alpha\}$ has the finite intersection property. Weak$^*$
compactness of $\mathcal S(\A)$ therefore gives
$C_\infty:=\bigcap_\alpha C_\alpha\neq\varnothing$.

If $\omega,\eta\in C_\infty$ and $A\in\A_0$, weak$^*$ continuity of
$\rho\mapsto\rho(A)$ gives
\[
|\omega(A)-\eta(A)|
\leq
\operatorname{osc}_{C_\alpha}(A)
=
\operatorname{osc}_{O_\alpha}(A)
\longrightarrow0.
\]
Thus $\omega=\eta$ on $\A_0$, and norm density and continuity imply
$\omega=\eta$ on $\A$. Hence $C_\infty$ is a singleton.
\end{proof}

\begin{remark}
Closures are essential in Theorem~\ref{thm:reduction}: decreasing
non-empty open sets can have empty intersection, as
$O_n=(0,1/n)\subset\mathbb R$ shows, while
$\bigcap_n\overline{O_n}=\{0\}$. Thus increasingly precise parcels
determine the limiting state through their closures even when no state
belongs to every parcel.
\end{remark}
\section{Finite-Dimensional IAQFT and Information Gain}
\label{sec:finite}

Let $\dim\mathcal H=n<\infty$ and
$\mathcal A=B(\mathcal H)\cong M_n(\mathbb C)$. The canonical map
\[
\Theta:\mathcal D(\mathcal H)\longrightarrow\mathcal S(\mathcal A),
\qquad
\Theta(\rho)(A)=\operatorname{Tr}(\rho A),
\]
is an affine bijection. Since in finite dimensions the Euclidean and weak
topologies on $\mathcal D(\mathcal H)$ correspond to the weak$^*$ topology
on $\mathcal S(\mathcal A)$, $\Theta$ is an affine homeomorphism. Moreover,
it intertwines the L\"uders updates: if
\[
f_j(\rho)
=
\frac{M_j\rho M_j^*}{\operatorname{Tr}(\rho E_j)},
\]
then
\[
\Theta(f_j(\rho))(A)
=
\frac{\Theta(\rho)(M_j^*AM_j)}{\Theta(\rho)(E_j)}.
\]
Thus finite-dimensional IAQFT and IQM~\cite{Edalat2026} are canonically
equivalent at the level of states, parcels, their information order and
measurement updates. The corresponding finite-dimensional IQM results
therefore transfer directly to IAQFT and will not be restated here.

\begin{proposition}[Operational and Hilbert--Schmidt volume]
\label{prop:operational-HS}
Let $d=n^2-1$ and let $H_1,\ldots,H_d$ be a Hilbert--Schmidt
orthonormal basis of the real vector space of traceless self-adjoint
operators on $\mathcal H$. Set
\[
V=\operatorname{span}_{\mathbb R}\{H_1,\ldots,H_d\}.
\]
Then
\[
\Phi_V(\Theta(\rho))
=
\bigl(
\operatorname{Tr}(\rho H_1),\ldots,
\operatorname{Tr}(\rho H_d)
\bigr)
\]
is the restriction to $\mathcal S(\mathcal A)$ of an affine isometry
from the trace-one self-adjoint affine space, equipped with the
Hilbert--Schmidt metric, onto $\mathbb R^d$. Consequently, for every
parcel $O\subseteq\mathcal S(\mathcal A)$,
\[
\operatorname{Vol}_V(O)
=
\operatorname{Vol}_{\mathrm{HS}}\bigl(\Theta^{-1}(O)\bigr).
\]
For any other basis of the same observational space, the volumes are
multiplied by a fixed positive coordinate factor; hence volume ratios and
the corresponding relative operational information gains are unchanged.
\end{proposition}
\begin{proof}
If $\rho$ and $\sigma$ have trace one, then $\rho-\sigma$ is traceless,
and hence
\[
\rho-\sigma
=
\sum_{k=1}^d
\operatorname{Tr}\bigl((\rho-\sigma)H_k\bigr)H_k.
\]
Hilbert--Schmidt orthonormality therefore gives
\[
\|\rho-\sigma\|_{\mathrm{HS}}^2
=
\sum_{k=1}^d
\left|
\operatorname{Tr}\bigl((\rho-\sigma)H_k\bigr)
\right|^2
=
\|\Phi_V(\Theta(\rho))-\Phi_V(\Theta(\sigma))\|_{\mathbb R^d}^2.
\]
Thus the observational-coordinate map is an affine isometry and
preserves the corresponding Lebesgue volume. An invertible change of
observational coordinates multiplies all such volumes by the same
absolute determinant, so their ratios are unchanged.
\end{proof}

The preceding proposition identifies the operational volume of
Section~\ref{sec:operational-volume}, for an informationally complete
finite-dimensional observational space, with Hilbert--Schmidt volume.
Thus the volume-contraction results below are also quantitative
operational-information results. We now strengthen the corresponding
volume results of IQM~\cite{Edalat2026} by deriving the exact Jacobian for arbitrary
invertible L\"uders updates and a contraction theorem for arbitrary
finite-dimensional Hilbert spaces.

For notational convenience, we transport Hilbert--Schmidt volume from
$\mathcal D(\mathcal H)$ to $\mathcal S(\mathcal A)$ through $\Theta$ and
write
\[
\operatorname{Vol}(O)
:=
\operatorname{Vol}_{\mathrm{HS}}\bigl(\Theta^{-1}(O)\bigr)
=
\operatorname{Vol}_V(O)
\]
for the Hilbert--Schmidt orthonormal observational coordinates fixed
above.

\begin{definition}[$j$-positive parcel]
\label{def:jpositive}
For a POVM effect $E_j$, a parcel $O\subseteq\mathcal S(\mathcal A)$ is
\emph{$j$-positive} if
\[
p_{\min}:=\inf_{\omega\in O}\omega(E_j)>0.
\]
\end{definition}

The following elementary Jacobian results are proved in the Appendix,
Section~\ref{sec:proofs-vo}.

\begin{lemma}[Projective-linear Jacobian]
\label{lem:proj}
Let $\mathcal V$ be a real vector space of dimension $d+1$,
$\ell:\mathcal V\to\mathbb R$ a nonzero linear functional,
$W=\{x:\ell(x)=1\}$, and $L:\mathcal V\to\mathcal V$ invertible. For
$f(x)=Lx/\ell(Lx)$, wherever $\ell(Lx)>0$,
\[
|\det Df(x)|
=
\frac{|\det L|}{\ell(Lx)^{d+1}}.
\]
\end{lemma}

Let $\operatorname{Herm}(\mathcal H)$ denote the real vector space of
Hermitian operators on $\mathcal H$, of real dimension $n^2$.

\begin{lemma}[Determinant of the linear part]
\label{lem:detL}
For $L:\operatorname{Herm}(\mathcal H)\to
\operatorname{Herm}(\mathcal H)$ defined by $L(\rho)=M_j\rho M_j^*$,
where $E_j=M_j^*M_j$,
\[
\det_{\mathbb R}L=(\det E_j)^n.
\]
\end{lemma}

\begin{proposition}[Exact Jacobian of the L\"uders map]
\label{prop:jacobian}
If $M_j$ is invertible and $\rho\in\mathcal D(\mathcal H)$, then
\[
\left|\det Df_j(\rho)\right|
=
\frac{(\det E_j)^n}
     {\operatorname{Tr}(E_j\rho)^{n^2}}.
\]
\end{proposition}

This immediately yields an exact volume formula for arbitrary invertible
POVM effects.

\begin{theorem}[Exact volume formula and contraction criterion]
\label{thm:exactvolume}
Let $O\subseteq\mathcal S(\mathcal A)$ be a $j$-positive parcel and
suppose $M_j$ is invertible. Then
\[
\operatorname{Vol}(f_j(O))
=
\int_{\Theta^{-1}(O)}
\frac{(\det E_j)^n}
     {\operatorname{Tr}(E_j\rho)^{n^2}}
\,d\mu(\rho)
\leq
\frac{(\det E_j)^n}{p_{\min}^{\,n^2}}
\operatorname{Vol}(O).
\]
In particular,
$\operatorname{Vol}(f_j(O))<\operatorname{Vol}(O)$ whenever
$(\det E_j)^{1/n}<p_{\min}$.
\end{theorem}
\begin{proof}
Under the identification $\Theta$, the state-space update is represented
on $\mathcal D(\mathcal H)$ by
$\rho\mapsto M_j\rho M_j^*/\operatorname{Tr}(\rho E_j)$.
Since $M_j$ is invertible, this map is a smooth diffeomorphism onto its
image. The change-of-variables formula and
Proposition~\ref{prop:jacobian} give the equality. Since
$\operatorname{Tr}(E_j\rho)\geq p_{\min}$ on $\Theta^{-1}(O)$, the
stated bound follows, and it is strict when
$(\det E_j)^n<p_{\min}^{n^2}$.
\end{proof}

\begin{example}[Qubit volume contraction]
\label{ex:qubit-volume}
Let $\mathcal H=\mathbb C^2$ and consider an invertible POVM effect
\[
E=
\begin{pmatrix}
\lambda_1&0\\
0&\lambda_2
\end{pmatrix},
\qquad
0<\lambda_1,\lambda_2\leq1.
\]
For the corresponding L\"uders update, Proposition~\ref{prop:jacobian}
gives
\[
\left|\det Df(\rho)\right|
=
\frac{(\lambda_1\lambda_2)^2}
     {\operatorname{Tr}(E\rho)^4}.
\]
If a parcel $O\subseteq\mathcal S(M_2(\mathbb C))$ satisfies
\[
p_{\min}
:=
\inf_{\rho\in\Theta^{-1}(O)}
\operatorname{Tr}(E\rho)>0,
\]
Theorem~\ref{thm:exactvolume} yields
\[
\operatorname{Vol}(f(O))
\leq
\frac{(\lambda_1\lambda_2)^2}{p_{\min}^4}
\operatorname{Vol}(O).
\]
Hence the update strictly contracts Hilbert--Schmidt volume whenever
\[
\sqrt{\lambda_1\lambda_2}<p_{\min}.
\]

For example, take $\lambda_1=0.9$, $\lambda_2=0.1$ and a parcel for
which $p_{\min}=0.4$. Then
$\sqrt{\lambda_1\lambda_2}=0.3<0.4$, and
\[
\operatorname{Vol}(f(O))
\leq
\frac{(0.09)^2}{(0.4)^4}\operatorname{Vol}(O)
=
\frac{81}{256}\operatorname{Vol}(O)
\approx0.316\,\operatorname{Vol}(O).
\]
Thus this finite-resolution measurement reduces the parcel volume by
at least about $68\%$.
\end{example}

 The contraction factor separates the measurement-dependent quantity
$(\det E_j)^{1/n}$ from the parcel-dependent quantity $p_{\min}$. Thus
the relative operational information gain produced by a contracting
update depends both on the measurement and on the position of the parcel:
approaching the zero-probability boundary
$\{\rho:\operatorname{Tr}(E_j\rho)=0\}$ increases the normalization in
the L\"uders update and can counteract geometric contraction.
The next result extends the fuzzy-projective contraction results of
IQM, which were obtained for qubit and multi-qubit systems, to arbitrary
finite-dimensional Hilbert spaces.
\begin{theorem}[Volume contraction for fuzzy projective measurements]
\label{thm:volumecontraction}
Let $\dim\mathcal H=n$, let $\Pi$ be a nonzero proper projection, and for
$0<\eta<1$ set
\[
E^{(\eta)}
=
\eta\Pi+\frac{1-\eta}{n}I,
\qquad
M^{(\eta)}=(E^{(\eta)})^{1/2}.
\]
If $O\subseteq\mathcal S(\mathcal A)$ is a parcel satisfying
\[
c:=
\inf_{\rho\in\Theta^{-1}(O)}
\operatorname{Tr}(\rho\Pi)>0,
\]
then there exists $\eta_0<1$ such that
\[
\operatorname{Vol}(f^{(\eta)}(O))
<
\operatorname{Vol}(O)
\qquad(\eta_0<\eta<1).
\]
\end{theorem}

\begin{proof}
Let $r=\operatorname{rank}\Pi$, so $1\leq r<n$. The eigenvalues of
$E^{(\eta)}$ are
$\alpha_\eta=\eta+(1-\eta)/n$ with multiplicity $r$ and
$\beta_\eta=(1-\eta)/n$ with multiplicity $n-r$. Hence
$\det E^{(\eta)}=\alpha_\eta^r\beta_\eta^{\,n-r}\to0$ as
$\eta\to1^-$, while for $\rho\in\Theta^{-1}(O)$,
$\operatorname{Tr}(E^{(\eta)}\rho)\geq\alpha_\eta c\to c>0$.
Proposition~\ref{prop:jacobian} therefore gives
\[
\sup_{\rho\in\Theta^{-1}(O)}
\left|\det Df^{(\eta)}(\rho)\right|
\leq
\frac{(\det E^{(\eta)})^n}
     {(\alpha_\eta c)^{n^2}}
\longrightarrow0.
\]
Thus the Jacobian is uniformly less than one for all $\eta$ sufficiently
close to $1$, and the change-of-variables formula gives the claimed strict
volume contraction.
\end{proof}

\begin{remark}
If successive selective measurements satisfy
$\operatorname{Vol}(O_{k+1})\leq c_k\operatorname{Vol}(O_k)$ with
$0<c_k\leq1$, then
$\operatorname{Vol}(O_k)\leq
(\prod_{i=0}^{k-1}c_i)\operatorname{Vol}(O_0)$. Hence the volumes tend to
zero whenever $\prod_{k=0}^\infty c_k=0$, a weaker requirement than a
uniform bound $c_k\leq c<1$.
\end{remark}

Thus finite-dimensional IAQFT inherits the parcel and measurement theory
of IQM but strengthens its quantitative geometry: the exact Jacobian and
volume formula apply to arbitrary invertible POVM effects, while the
fuzzy-projective contraction theorem applies to arbitrary
finite-dimensional Hilbert spaces rather than only qubit tensor-product
systems.
\section{Nets and Local Reduction}\label{sec:nets-germs}

The reduction results of Section~\ref{sec:reduction} apply to the
quasi-local algebra directly. Locality enters IAQFT by assigning compatible
parcels to the local algebras of the AQFT net.

Let $\mathcal K$ denote the bounded spacetime regions, directed by
inclusion, and let $\{\A(\mathcal O)\}_{\mathcal O\in\mathcal K}$ be the
local net. For $\mathcal O_1\subseteq\mathcal O_2$, let
\[
\iota_{12}:\A(\mathcal O_1)\longrightarrow\A(\mathcal O_2)
\]
denote the injective unital $*$-homomorphism that identifies each
observable of the smaller region $\mathcal O_1$ with the corresponding
observable in the larger region $\mathcal O_2$, as required by isotony.
It induces the restriction map on state spaces
\[
\iota_{12}^*:
\mathcal S(\A(\mathcal O_2))
\longrightarrow
\mathcal S(\A(\mathcal O_1)),
\qquad
\iota_{12}^*(\omega)=\omega\circ\iota_{12}.
\]

\begin{definition}[Parcel net]
\label{def:parcelnet}
A \emph{parcel net} assigns to each $\mathcal O\in\mathcal K$ a parcel
$\mathcal P(\mathcal O)\subseteq\mathcal S(\A(\mathcal O))$ such that,
whenever $\mathcal O_1\subseteq\mathcal O_2$,
\[
\iota_{12}^*\bigl(\mathcal P(\mathcal O_2)\bigr)
\subseteq
\mathcal P(\mathcal O_1).
\]
Thus every state compatible with the parcel information on the larger
region remains compatible with the parcel information specified on the
smaller region after restriction. Equality is not required, so
compatibility under restriction does not preclude additional
observational information at different stages.
\end{definition}

Operationally only finitely many regions and observational constraints are
available at any stage; the compatibility condition expresses when such
finite local data can be regarded as part of a common parcel net.

Let $\{\mathcal O_\alpha\}_{\alpha\in I}$ be a decreasing neighbourhood
basis of a point $x$. Identifying the local algebras with subalgebras of
the quasi-local algebra, set
\[
\A_x:=\bigcap_{\alpha\in I}\A(\mathcal O_\alpha).
\]
Since all local algebras contain the same identity,
$\mathbb C1\subseteq\A_x$, although $\A_x$ may equal $\mathbb C1$.
A family of states $\{\omega_\alpha\}$, with
$\omega_\alpha\in\mathcal S(\A(\mathcal O_\alpha))$, is called
\emph{compatible} if
$\omega_\beta=\iota_{\beta\alpha}^*(\omega_\alpha)$ whenever
$\mathcal O_\beta\subseteq\mathcal O_\alpha$. Such a compatible family
induces a state on $\A_x$ by $\omega_x(A)=\omega_\alpha(A)$;
compatibility makes this independent of $\alpha$. Global states provide
compatible families by restriction, although distinct compatible families
may induce the same state on $\A_x$.

Parcel-net compatibility and observational refinement are distinct
requirements. The former expresses consistency of finite-precision
descriptions under restriction between local algebras; it does not by
itself imply that the information induced on a limiting algebra becomes
more precise. Local reduction requires, in addition, that the induced
closed state sets on the limiting algebra form a decreasing family whose
observational oscillations vanish.

\begin{theorem}[Local Reduction Theorem]
\label{thm:localreduction}
Let $\{\mathcal O_\alpha\}_{\alpha\in I}$ be the decreasing neighbourhood
basis of $x$ introduced above, and let
$\mathcal P(\mathcal O_\alpha)$, for $\alpha\in I$,  be a compatible family of parcels in the
sense of Definition~\ref{def:parcelnet}. Set
\[
C_\alpha
=
\overline{
r_\alpha(\mathcal P(\mathcal O_\alpha))
}^{\,w^*},
\]
where $r_\alpha$ denotes restriction to $\A_x$, for $\alpha\in I$. Assume that $\{C_\alpha\}_{\alpha\in I}$ is decreasing and that, for
every $A$ in a norm-dense $*$-subalgebra
$\A_0\subseteq\A_x$,
\[
\lim_{\alpha\in I}
\operatorname{osc}_{\mathcal P(\mathcal O_\alpha)}(A)
=
0.
\]
Then
\[
\bigcap_{\alpha\in I} C_\alpha=\{\omega_x\}
\]
for a unique state $\omega_x\in\mathcal S(\A_x)$. Moreover, every compatible family
$\{\omega_\alpha\}_{\alpha\in I}$ satisfying
$\omega_\alpha\in\mathcal P(\mathcal O_\alpha)$ for every
$\alpha\in I$ induces this same state on $\A_x$.

\end{theorem}

\begin{proof}
Each $C_\alpha$ is non-empty and weak$^*$ compact, being a non-empty
closed subset of the weak$^*$ compact state space $\mathcal S(\A_x)$.
Since the family is decreasing, compactness gives
$C_\infty
:=
\bigcap_{\alpha\in I}C_\alpha
\neq\varnothing$. If $\omega,\eta\in C_\infty$ and $A\in\A_0$, then for every
$\alpha\in I$, weak$^*$ continuity of evaluation gives
\[
|\omega(A)-\eta(A)|
\leq
\operatorname{osc}_{C_\alpha}(A)
=
\operatorname{osc}_{\mathcal P(\mathcal O_\alpha)}(A).
\]
Taking the limit along the directed set $I$ and using the hypothesis,
\[
\lim_{\alpha\in I}
\operatorname{osc}_{\mathcal P(\mathcal O_\alpha)}(A)
=
0,
\]
we obtain
\[
|\omega(A)-\eta(A)|=0.
\]
Thus $\omega=\eta$ on $\A_0$, hence on $\A_x$ by norm density and
continuity. Finally, if
$\{\omega_\alpha\}$ is a compatible family with
$\omega_\alpha\in\mathcal P(\mathcal O_\alpha)$, for $\alpha\in I$, its restrictions to
$\A_x$ coincide by compatibility; denote the resulting state by
$\omega_x'$. Since
$r_\alpha(\omega_\alpha)\in
r_\alpha(\mathcal P(\mathcal O_\alpha))\subseteq C_\alpha$
for every $\alpha\in I$, we have
$\omega_x'\in
\bigcap_{\alpha\in I}C_\alpha=\{\omega_x\}$.
Hence $\omega_x'=\omega_x$.
\end{proof}

For spacelike separated regions $\mathcal O_1,\mathcal O_2$, Einstein
causality gives
$[\A(\mathcal O_1),\A(\mathcal O_2)]=0$. This yields the usual
no-signalling property directly at parcel level.

\begin{proposition}[Parcel no-signalling]
\label{prop:nosignalling}
Let $\{M_j\}\subseteq\A(\mathcal O_1)$ satisfy
$\sum_jM_j^*M_j=I$, with $\mathcal O_1$ spacelike separated from
$\mathcal O_2$, and let
$\mathcal E^*(\omega)(A)=\sum_j\omega(M_j^*AM_j)$ be the nonselective
update. Then, for every parcel $P$ and every $A\in\A(\mathcal O_2)$,
\[
\{\mathcal E^*(\omega)(A):\omega\in P\}
=
\{\omega(A):\omega\in P\}.
\]
Hence all finite-precision expectation intervals in $\mathcal O_2$ are
unchanged by a nonselective measurement in $\mathcal O_1$.
\end{proposition}

\begin{proof}
Since $A$ commutes with every $M_j$,
\[
\mathcal E^*(\omega)(A)
=
\sum_j\omega(AM_j^*M_j)
=
\omega(A)
\]
for every state $\omega$. The assertion follows pointwise over $P$.
\end{proof}

A selective update may change the corresponding expectation intervals
because conditioning can exploit correlations between spacelike separated
regions. This is not signalling: without knowledge of the selected outcome,
the relevant description is the nonselective update, for which
Proposition~\ref{prop:nosignalling} gives exact invariance.
\section{Bell Correlations and Geometric Reduction}
\label{sec:Bell-correlation}

Throughout this and the following sections, let $\omega_0$ denote the
distinguished vacuum state of the quasi-local algebra $\A$, with GNS
representation $(\pi_0,\mathcal H_0,\Omega_0)$, so that
$\omega_0(A)=\langle\Omega_0,\pi_0(A)\Omega_0\rangle$. IAQFT retains the
usual AQFT vacuum; finite-precision information about it is represented
by parcels containing $\omega_0$.

Let $\mathcal O_1,\mathcal O_2$ be spacelike separated. For
$A\in\A(\mathcal O_1)$ and $B\in\A(\mathcal O_2)$, a nonzero vacuum
correlation
$\omega_0(AB)-\omega_0(A)\omega_0(B)$ is a property of the state, not a
signal between the regions. Selective measurement may change conditional
information in the remote region because of such correlations, whereas
the nonselective update leaves its expectation values unchanged by
Proposition~\ref{prop:nosignalling}.

\begin{definition}[Correlation interval]
\label{def:correlationinterval}
For a parcel $O$ define
\[
C_O(A,B)
=
\{\omega(AB)-\omega(A)\omega(B):\omega\in O\}.
\]
\end{definition}

\begin{theorem}[Persistence of spacelike correlations]
\label{thm:vacuumcorrelation}
Let $A\in\A(\mathcal O_1)$ and $B\in\A(\mathcal O_2)$ be self-adjoint and
set $c=\omega_0(AB)-\omega_0(A)\omega_0(B)\neq0$. Then there is an observational parcel
$O\ni\omega_0$ such that every $\omega\in O$ satisfies
\[
|\omega(AB)-\omega(A)\omega(B)-c|<\frac{|c|}{2}.
\]
In particular, the correlation has the same sign as $c$ and magnitude
greater than $|c|/2$ throughout $O$.
\end{theorem}

\begin{proof}
The map $F(\omega)=\omega(AB)-\omega(A)\omega(B)$ is weak$^*$ continuous
and $F(\omega_0)=c$. Hence
\[
\{\omega:|F(\omega)-c|<|c|/2\}
\]
is a weak$^*$ neighbourhood of $\omega_0$ and therefore contains an
observational parcel containing $\omega_0$. Since every observational
parcel is a parcel, the conclusion follows.
\end{proof}

Strict Bell nonclassicality is similarly stable under finite precision.

\begin{theorem}[Persistence of Bell violation]
\label{thm:bellviolation}
Let $A_1,A_2\in\A(\mathcal O_1)$ and
$B_1,B_2\in\A(\mathcal O_2)$ be self-adjoint contractions and define
$\mathcal B=A_1(B_1+B_2)+A_2(B_1-B_2)$. If
$|\omega_0(\mathcal B)|>2$, then there is an observational parcel $O\ni\omega_0$ such
that $|\omega(\mathcal B)|>2$ for every $\omega\in O$.
\end{theorem}

\begin{proof}
Set $\varepsilon=(|\omega_0(\mathcal B)|-2)/2>0$. The set
\[
O=
\{\omega\in\mathcal S(\A):
|\omega(\mathcal B)-\omega_0(\mathcal B)|<\varepsilon\}
\]
is an observational parcel containing $\omega_0$, and for
$\omega\in O$,
\[
|\omega(\mathcal B)|
>
|\omega_0(\mathcal B)|-\varepsilon
=
2+\varepsilon>2.
\]
\end{proof}

\begin{remark}
Theorem~\ref{thm:bellviolation} is conditional on a vacuum CHSH violation
for the specified local observables; it does not assert their existence.
Strong Bell-correlation results for spacelike-separated local algebras are
known under suitable AQFT hypotheses; see, for example,
Summers--Werner~\cite{SummersWerner}. Whenever such a strict violation occurs, the theorem shows that it
persists throughout a sufficiently small observational parcel.
\end{remark}

These results give a geometric interpretation of reduction without
weakening quantum nonlocality. Spacelike correlations, including strict
Bell violations, persist on open regions of state space and are therefore
not artefacts of infinitely precise point states. A selective measurement
updates the parcel of admissible states and can thereby change conditional
information about a spacelike-separated region, while the nonselective
update leaves its local expectation values unchanged. Thus IAQFT retains
the nonlocal correlations of AQFT while representing collapse as a
geometric transformation of finite-precision state information rather
than as a superluminal physical influence.
\section{Reeh--Schlieder as Local Parcel Reachability}
\label{sec:Reeh-Schlieder}
Let $(\pi_0,\mathcal H_0,\Omega_0)$ be the GNS representation of the
vacuum state $\omega_0$, and assume that $\omega_0$ is pure, equivalently
that $\pi_0$ is irreducible. Hence
\[
\mathcal M_0:=\pi_0(\A)''=B(\mathcal H_0).
\]
We regard $\omega_0$ also as the vector state
$\omega_0(C)=\langle\Omega_0,C\Omega_0\rangle$ on $\mathcal M_0$.

Recall that a state on a von Neumann algebra $\mathcal M$ is
\emph{normal} if it belongs to the predual $\mathcal M_*$, equivalently
if it is weak$^*$ (ultraweakly) continuous on $\mathcal M$. The normal
state space
$\Sn(\mathcal M_0)\subset\mathcal M_{0*}$ is equipped with the weak
topology $\sigma(\mathcal M_{0*},\mathcal M_0)$. Its pure normal states
are precisely the vector states
\[
\omega_\Psi(C)=\langle\Psi,C\Psi\rangle
\]
for unit $\Psi\in\mathcal H_0$.

For $B\in\pi_0(\A(\mathcal O))$ with
$\omega_0(B^*B)>0$, define the local selective update
\[
(\omega_0)_B(C)
=
\frac{\omega_0(B^*CB)}{\omega_0(B^*B)}.
\]
It is the vector state induced by
$B\Omega_0/\|B\Omega_0\|$. After rescaling $B$ if necessary,
$B^*B\leq I$, so this update is the first outcome of the POVM
$\{B^*B,I-B^*B\}$.

\begin{proposition}[State-space Reeh--Schlieder property]
\label{prop:RS-state}
For a non-empty bounded open region $\mathcal O$, the following are
equivalent:
\begin{enumerate}
\item $\Omega_0$ is cyclic for $\pi_0(\A(\mathcal O))$;
\item
\[
\mathcal R_{\mathcal O}(\omega_0)
:=
\{(\omega_0)_B:
B\in\pi_0(\A(\mathcal O)),\ \omega_0(B^*B)>0\}
\]
is dense in the pure normal state space of $\mathcal M_0$ with respect
to the relative topology induced by
$\sigma(\mathcal M_{0*},\mathcal M_0)$.
\end{enumerate}
\end{proposition}

\begin{proof}
If $\Omega_0$ is cyclic and $\Psi\in\mathcal H_0$ is a unit vector,
choose the net $B_\alpha\in\pi_0(\A(\mathcal O))$, $\alpha\in I$, such that
$
B_\alpha\Omega_0\longrightarrow\Psi
$
in Hilbert-space norm. Then
\[
\Phi_\alpha
:=
\frac{B_\alpha\Omega_0}{\|B_\alpha\Omega_0\|}
\longrightarrow\Psi
\]
in Hilbert-space norm. For every $C\in\mathcal M_0$,
\[
(\omega_0)_{B_\alpha}(C)
=
\langle\Phi_\alpha,C\Phi_\alpha\rangle,
\qquad
\omega_\Psi(C)
=
\langle\Psi,C\Psi\rangle.
\]
Hence
\[
\begin{aligned}
\left|
(\omega_0)_{B_\alpha}(C)-\omega_\Psi(C)
\right|
&\leq
\left|
\langle\Phi_\alpha-\Psi,C\Phi_\alpha\rangle
\right|
+
\left|
\langle\Psi,C(\Phi_\alpha-\Psi)\rangle
\right|\\
&\leq
2\|C\|\,\|\Phi_\alpha-\Psi\|
\longrightarrow0.
\end{aligned}
\]
Thus
$
(\omega_0)_{B_\alpha}
\longrightarrow
\omega_\Psi
$
in the weak topology
$\sigma(\mathcal M_{0*},\mathcal M_0)$.

Conversely, suppose $\mathcal R_{\mathcal O}(\omega_0)$ is dense and let
$\Psi$ be a unit vector. For $\varepsilon>0$, the set
\[
V_{\Psi,\varepsilon}
=
\{\omega\in\Sn(\mathcal M_0):
\omega(P_\Psi)>1-\varepsilon\},
\qquad
P_\Psi=|\Psi\rangle\langle\Psi|,
\]
is a neighbourhood of $\omega_\Psi$. Hence some
$(\omega_0)_{B_\varepsilon}\in V_{\Psi,\varepsilon}$.Writing
\[
\Phi_\varepsilon
=
\frac{B_\varepsilon\Omega_0}
{\|B_\varepsilon\Omega_0\|},
\]
we have
$
|\langle\Psi,\Phi_\varepsilon\rangle|^2
=
(\omega_0)_{B_\varepsilon}(P_\Psi)
>
1-\varepsilon.
$
Multiplying $B_\varepsilon$ by a scalar of modulus one if necessary,
which leaves $(\omega_0)_{B_\varepsilon}$ unchanged, we may assume that
$\langle\Psi,\Phi_\varepsilon\rangle$ is real and nonnegative. Hence
\[
\|\Psi-\Phi_\varepsilon\|^2
=
2-2|\langle\Psi,\Phi_\varepsilon\rangle|
<
2-2\sqrt{1-\varepsilon}
\longrightarrow0.
\]
Thus
$\Psi\in\overline{\pi_0(\A(\mathcal O))\Omega_0}$, and cyclicity follows.\end{proof}
\begin{definition}[Local parcel reachability]
\label{def:localreachability}
A normal parcel $V\subseteq\Sn(\mathcal M_0)$ is \emph{locally reachable
from the vacuum through $\mathcal O$} if
$(\omega_0)_B\in V$ for some
$B\in\pi_0(\A(\mathcal O))$ with $\omega_0(B^*B)>0$.
\end{definition}
The terminology is analogous to reachability in control theory
\cite[Ch.~3]{sontag2013mathematical}: the vacuum plays the role of the
initial state, local operations in $\A(\mathcal O)$ are the admissible
controls, and the target is a finite-precision parcel rather than an
exact state.
\begin{example}[Finite-precision local target]
Let $\Psi\in\mathcal H_0$ be a unit vector and let $V$ be a normal parcel
containing the vector state $\omega_\Psi$. If the Reeh--Schlieder
property holds for $\mathcal O$, Proposition~\ref{prop:RS-state} gives
some $B\in\pi_0(\A(\mathcal O))$ such that $(\omega_0)_B\in V$.
Thus a local operation need not prepare $\omega_\Psi$ exactly: it can
prepare a state satisfying any prescribed finite-precision neighbourhood
of that target.
\end{example}
\begin{theorem}[Reeh--Schlieder as local parcel reachability]
\label{thm:reehschlieder}
Under the assumptions above, the Reeh--Schlieder property for
$\mathcal O$ holds if and only if every normal parcel containing a pure
normal state is locally reachable from the vacuum through $\mathcal O$.
\end{theorem}

\begin{proof}
By Proposition~\ref{prop:RS-state}, Reeh--Schlieder is equivalent to
density of $\mathcal R_{\mathcal O}(\omega_0)$ in the pure normal state
space. Density immediately implies that every open parcel containing a
pure state meets $\mathcal R_{\mathcal O}(\omega_0)$.

Conversely, if every such parcel is reachable, let $\varphi$ be a pure
normal state and $U$ any neighbourhood of $\varphi$. Local convexity gives
a convex open neighbourhood $V$ with $\varphi\in V\subseteq U$.
Reachability gives
$V\cap\mathcal R_{\mathcal O}(\omega_0)\neq\varnothing$, hence also
$U\cap\mathcal R_{\mathcal O}(\omega_0)\neq\varnothing$. Thus
$\mathcal R_{\mathcal O}(\omega_0)$ is dense, and
Proposition~\ref{prop:RS-state} gives Reeh--Schlieder.
\end{proof}

Thus the classical statement that local operations on the vacuum generate
a dense set of vectors becomes, at finite precision, the statement that
every parcel around a pure normal target state is reachable by a suitable
local operation. This clarifies the apparently counterintuitive
``remote preparation'' aspect of the Reeh--Schlieder theorem: its
finite-precision operational content does not require exact preparation
of an arbitrarily specified global target state, but reachability of every
finite-precision neighbourhood of such a target. The underlying
Reeh--Schlieder theorem and its compatibility with relativistic causality
are unchanged; IAQFT makes explicit its operational content when target
states can be specified only to finite precision.
\section{Haag's Theorem and Parcel Equivalence}
\label{sec:haag}

Haag-type results obstruct vacuum-preserving unitary equivalence between
the local nets of genuinely interacting and free quantum field theories.
The corresponding finite-precision question is whether such an
obstruction can disappear when exact vacuum states are replaced by parcel
refinements. We show that it cannot: parcel equivalence already implies
equivalence of the limiting vacuum representations.

\begin{definition}[Vacuum refinement family]
A \emph{refinement family} on a unital $C^*$-algebra $\A$ is a decreasing
family of parcels $\mathcal V=\{O_\alpha\}_{\alpha\in I}$ for which there
exists a norm-dense $*$-subalgebra $\A_0\subseteq\A$ such that, for every
$A\in(\A_0)_{\rm sa}:=\{A\in\A_0:A=A^*\}$,
\[
\lim_{\alpha\in I}
\operatorname{osc}_{O_\alpha}(A)
=
0.
\]
By Theorem~\ref{thm:reduction},
$\bigcap_{\alpha\in I}\overline{O_\alpha}^{\,w^*}$ consists of a unique
state, denoted $\omega_{\mathcal V}$. If this is the vacuum state,
$\mathcal V$ is called a \emph{vacuum refinement family}.
\end{definition}

\begin{definition}[Parcel equivalence]
\label{def:parceleq}
Let $\mathfrak A_1,\mathfrak A_2$ be AQFT nets with quasi-local algebras
$\A_1,\A_2$ and vacuum refinement families
$\mathcal V_1=\{O_\alpha^1\}_{\alpha\in I}$ and
$\mathcal V_2=\{O_\beta^2\}_{\beta\in J}$. They are
\emph{parcel equivalent} if there is a $*$-isomorphism
$\Phi:\A_1\to\A_2$ satisfying
$\Phi(\A_1(\mathcal O))=\A_2(\mathcal O)$ for every region $\mathcal O$
and such that the transported family
\[
\Phi_*(\mathcal V_1)
=
\{\Phi_*(O_\alpha^1)\}_{\alpha\in I},
\qquad
\Phi_*(\omega)=\omega\circ\Phi^{-1},
\]
is cofinal with $\mathcal V_2$ under reverse inclusion:
for every $\alpha\in I$ some $\beta\in J$ satisfies
$O_\beta^2\subseteq\Phi_*(O_\alpha^1)$, and for every $\beta\in J$ some
$\alpha\in I$ satisfies
$\Phi_*(O_\alpha^1)\subseteq O_\beta^2$.
\end{definition}

Parcel equivalence is stronger than equality of the limiting vacuum
states: it compares their complete finite-precision refinement structures.

\begin{lemma}[Cofinal refinement families have the same limit]
\label{lem:cofinal-refinements}
If $\{O_\alpha\}$ and $\{V_\beta\}$ are cofinal under reverse inclusion,
then
\[
\bigcap_\alpha\overline{O_\alpha}^{\,w^*}
=
\bigcap_\beta\overline{V_\beta}^{\,w^*}.
\]
Consequently, if the two intersections are singletons, their limiting
states coincide.
\end{lemma}

\begin{proof}
For fixed $\beta$, choose $\alpha$ with $O_\alpha\subseteq V_\beta$;
then $\overline{O_\alpha}\subseteq\overline{V_\beta}$. Hence every point
of $\bigcap_\alpha\overline{O_\alpha}^{\,w^*}$ lies in every
$\overline{V_\beta}$. The reverse inclusion follows identically from the
other cofinality condition.
\end{proof}

\begin{theorem}[Parcel equivalence implies vacuum GNS equivalence]
\label{thm:parceleq}
Let $(\mathfrak A_1,\mathcal V_1)$ and
$(\mathfrak A_2,\mathcal V_2)$ be parcel-equivalent AQFT models, with
limiting vacuum states $\omega_1,\omega_2$, and let
$\Phi:\A_1\to\A_2$ implement the parcel equivalence. Then
\[
\omega_1=\omega_2\circ\Phi.
\]
Consequently their vacuum GNS representations are unitarily equivalent:
there is a unitary $U:\mathcal H_1\to\mathcal H_2$ such that
\[
U\Omega_1=\Omega_2,
\qquad
U\pi_1(A)U^{-1}=\pi_2(\Phi(A)),
\quad A\in\A_1,
\]
and therefore
\[
U\pi_1(\A_1(\mathcal O))U^{-1}
=
\pi_2(\A_2(\mathcal O))
\]
for every region $\mathcal O$.
\end{theorem}

\begin{proof}
The induced state-space homeomorphism
$\Phi_*(\omega)=\omega\circ\Phi^{-1}$ sends the limiting state
$\omega_1$ of $\mathcal V_1$ to the limiting state
$\omega_1\circ\Phi^{-1}$ of $\Phi_*(\mathcal V_1)$. By parcel
equivalence, $\Phi_*(\mathcal V_1)$ and $\mathcal V_2$ are cofinal, so
Lemma~\ref{lem:cofinal-refinements} gives
$\omega_1\circ\Phi^{-1}=\omega_2$, equivalently
$\omega_1=\omega_2\circ\Phi$.

The usual uniqueness of the GNS construction now gives a unitary $U$
with $U\Omega_1=\Omega_2$ and
$U\pi_1(A)U^{-1}=\pi_2(\Phi(A))$. Since $\Phi$ preserves the local net,
the final assertion follows.
\end{proof}

The converse need not hold: two refinement families may converge to the
same vacuum state without being cofinal. Parcel equivalence therefore
retains finite-precision information that is absent from equivalence of
the limiting vacuum representations.

\begin{corollary}[Parcel form of a Haag-type obstruction]
\label{thm:haag}
If, under the hypotheses of a Haag-type result, two AQFT models cannot be
related by a vacuum-preserving unitary equivalence of their local nets,
then they cannot be parcel equivalent.
\end{corollary}

\begin{proof}
Parcel equivalence would imply precisely such a vacuum-preserving unitary
equivalence by Theorem~\ref{thm:parceleq}.
\end{proof}

This clarifies what finite precision does, and does not, change about
Haag's theorem. One might hope that replacing exact vacuum states by
finite-precision neighbourhoods could remove an obstruction that is
formulated at the level of exact representations. The result shows that
this is not so for parcel equivalence as defined here: equivalence of the
complete finite-precision vacuum refinement structures already forces
equivalence of their limiting vacuum GNS representations. Haag-type
obstructions therefore survive this strong notion of finite-precision
equivalence. Whether weaker notions based only on finitely many
observations or finite stages of refinement can be useful remains a
separate question.

\section{Tomita--Takesaki Theory in IAQFT}
\label{sec:modular}
Let $\mathcal M$ be a von Neumann algebra and let $\omega$ be a faithful
normal state, meaning that
\[
\omega(A^*A)=0
\quad\Longrightarrow\quad
A=0
\qquad(A\in\mathcal M).
\]
In the GNS representation
$(\pi_\omega,\mathcal H_\omega,\Omega_\omega)$, the vector
$\Omega_\omega$ is cyclic and separating for $\pi_\omega(\mathcal M)$.
The Tomita operator is the closure of the densely defined antilinear
operator
\[
S_\omega\,
\pi_\omega(A)\Omega_\omega
=
\pi_\omega(A)^*\Omega_\omega,
\qquad A\in\mathcal M.
\]
Its polar decomposition is
\[
S_\omega
=
J_\omega\Delta_\omega^{1/2},
\]
where $\Delta_\omega$ is the positive self-adjoint modular operator and
$J_\omega$ is the antiunitary modular conjugation. The modular operator
generates the modular automorphism group through
\[
\pi_\omega\!\left(\sigma_t^\omega(A)\right)
=
\Delta_\omega^{it}\pi_\omega(A)\Delta_\omega^{-it},
\qquad
A\in\mathcal M,\quad t\in\mathbb R,
\]
while modular conjugation exchanges the represented algebra with its
commutant:
\[
J_\omega\pi_\omega(\mathcal M)J_\omega
=
\pi_\omega(\mathcal M)'.
\]
Thus $\Delta_\omega$ determines the canonical modular dynamics associated
with the state $\omega$, whereas $J_\omega$ implements the corresponding
algebra--commutant relation. We examine how these structures act on
finite-precision state information.

For each $t$, modular evolution acts on states by pullback,
$(\sigma_t^{\omega *}\rho)(A)=\rho(\sigma_t^\omega(A))$.

\begin{proposition}[Modular evolution of parcels]
\label{prop:modularliftsparcel}
For every $t\in\mathbb R$, $\sigma_t^{\omega *}$ is an affine weak$^*$
homeomorphism of $\mathcal S(\mathcal M)$ which preserves
$\mathcal S^{\sf n}(\mathcal M)$. Consequently it maps parcels to parcels
and induces a one-parameter group of parcel homeomorphisms.
\end{proposition}

\begin{proof}
Pullback by the normal $*$-automorphism $\sigma_t^\omega$ preserves states,
normality and convex combinations. It is weak$^*$ continuous because
$\rho_\lambda\to\rho$ implies
$\rho_\lambda(\sigma_t^\omega(A))\to\rho(\sigma_t^\omega(A))$ for every
$A\in\mathcal M$, and its inverse is $\sigma_{-t}^{\omega *}$. Hence it
maps non-empty open convex sets to non-empty open convex sets.
\end{proof}

Modular evolution is also covariant with selective measurement. Let $f_j$
be the update associated with $E_j=M_j^*M_j$ and define
$\widetilde M_j=\sigma_{-t}^\omega(M_j)$,
$\widetilde E_j=\sigma_{-t}^\omega(E_j)$, with corresponding update
$\widetilde f_j$.

\begin{proposition}[Modular covariance of measurement]
\label{prop:covariancemodular}
For every $\rho$ with $\rho(E_j)>0$,
\[
\sigma_t^{\omega *}(f_j(\rho))
=
\widetilde f_j(\sigma_t^{\omega *}\rho).
\]
The same identity holds pointwise for parcels. If $M_j$ belongs to the
centralizer
\[
\mathcal M_\omega
=
\{A\in\mathcal M:\sigma_t^\omega(A)=A\text{ for all }t\},
\]
then $\sigma_t^{\omega *}\circ f_j=f_j\circ\sigma_t^{\omega *}$.
\end{proposition}

\begin{proof}
For $A\in\mathcal M$,
\[
\begin{aligned}
\widetilde f_j(\sigma_t^{\omega *}\rho)(A)
&=
\frac{
(\sigma_t^{\omega *}\rho)
(\widetilde M_j^*A\widetilde M_j)}
{(\sigma_t^{\omega *}\rho)(\widetilde E_j)}
&&
\text{by definition of $\widetilde f_j$}
\\[1ex]
&=
\frac{
\rho\!\left(
\sigma_t^\omega(\widetilde M_j)^*
\sigma_t^\omega(A)
\sigma_t^\omega(\widetilde M_j)
\right)}
{\rho(\sigma_t^\omega(\widetilde E_j))}
&&
\text{by definition of $\sigma_t^{\omega *}$ \&
$*$-automorphism property}
\\[1ex]
&=
\frac{\rho(M_j^*\sigma_t^\omega(A)M_j)}
{\rho(E_j)}
&&
\text{since $\sigma_t^\omega(\widetilde M_j)=M_j$ and
$\sigma_t^\omega(\widetilde E_j)=E_j$}
\\[1ex]
&=
\sigma_t^{\omega *}(f_j(\rho))(A)
&&
\text{by definitions of $f_j$ and $\sigma_t^{\omega *}$.}
\end{aligned}
\]
The parcel identity follows pointwise, and the final assertion follows
when $\widetilde M_j=M_j$.
\end{proof}

The modular conjugation gives a different structure. Tomita--Takesaki
theory gives $J_\omega\mathcal M J_\omega=\mathcal M'$. Since
$B\mapsto J_\omega BJ_\omega$ is conjugate-linear, define instead the
linear $*$-anti-isomorphism
\[
\Theta_J:\mathcal M'\longrightarrow\mathcal M,
\qquad
\Theta_J(B)=J_\omega B^*J_\omega.
\]
For $\rho\in\mathcal S^{\sf n}(\mathcal M)$ set
$\rho^J:=\rho\circ\Theta_J$.

\begin{proposition}[Modular conjugation induces parcel duality]
\label{prop:Jparcel}
The map
\[
\Phi_J:\mathcal S^{\sf n}(\mathcal M)
\longrightarrow
\mathcal S^{\sf n}(\mathcal M'),
\qquad
\rho\longmapsto\rho^J,
\]
is an affine homeomorphism for the relative weak$^*$ topologies, whose
inverse is the same construction with $\mathcal M$ and $\mathcal M'$
interchanged. Consequently it gives a bijection between normal parcels on
$\mathcal M$ and on $\mathcal M'$.
\end{proposition}

\begin{proof}
The map $\Theta_J$ is a unital, positive, ultraweakly continuous
$*$-anti-isomorphism, so composition with it sends normal states on
$\mathcal M$ to normal states on $\mathcal M'$. It follows immediately
from
\[
\rho^J(B)=\rho(J_\omega B^*J_\omega)
\]
that $\Phi_J$ is affine and weak$^*$ continuous. Since $J_\omega^2=I$,
applying the construction twice gives $(\rho^J)^J=\rho$; hence the inverse
is continuous by the same argument. The resulting affine homeomorphism
maps normal parcels bijectively to normal parcels.
\end{proof}

The modular constructions above require a faithful normal reference state
but not a vacuum state. In particular, if $\mathcal M=\mathcal M(\mathcal
O)$ is a local algebra in the vacuum representation and Haag duality
$\mathcal M(\mathcal O)'=\mathcal M(\mathcal O')$ holds, the
$J_\omega$-duality gives a bijection between normal parcels on
$\mathcal M(\mathcal O)$ and $\mathcal M(\mathcal O')$. This is an
algebraic parcel duality, not a claim about measurement outcomes in the
causal complement; operational no-signalling remains governed by
Proposition~\ref{prop:nosignalling}. Section~\ref{sec:kms} applies the
modular framework to finite-precision KMS structure.
\section{KMS Condition}
\label{sec:kms}

We now formulate the KMS condition in finite-precision terms. The main
issue is that the modular dynamics $\sigma_t^\psi$ itself depends on
the state $\psi$. The relative-modular continuity result proved in
Appendix~\ref{sec:kms-continuity-proof} shows that this dependence is
continuous, in the sense required below, under predual-norm convergence
of faithful normal states. 

Let $\mathcal M$ be a $\sigma$-finite von Neumann algebra, meaning that
it admits a faithful normal state, and let
$\mathcal S_f(\mathcal M)$ denote its faithful normal states. This set is
non-empty and norm-dense in $\mathcal S^{\sf n}(\mathcal M)$
\cite[Prop.~III.2.2]{Takesaki1}; hence every normal parcel $O$ has its faithful part
$O_f:=O\cap\mathcal S_f(\mathcal M)\neq\varnothing$. For
$\psi\in O_f$, Tomita--Takesaki theory supplies the modular group
$\sigma_t^\psi$ and the usual KMS boundary values
\[
F_\psi^{A,B}(t)=\psi(A\sigma_t^\psi(B)),
\qquad
F_\psi^{A,B}(t+i)=\psi(\sigma_t^\psi(B)A).
\]

\begin{definition}[Modularly continuous parcel]
\label{def:modcontinuous}
A normal parcel $O$ is \emph{modularly continuous} if, for every
$A,B\in\mathcal M$ and $t\in\mathbb R$, the map
\[
\psi\longmapsto\psi(A\sigma_t^\psi(B))
\]
is continuous on $O_f$ in the predual norm.
\end{definition}

\begin{theorem}[Modular continuity of normal parcels]
\label{thm:modular-continuity}
Every normal parcel is modularly continuous.
\end{theorem}

\begin{proof}
Let $\psi_n,\omega\in O_f$ with
\[
\|\psi_n-\omega\|_{\mathcal M_*}\longrightarrow0.
\]
Lemma~\ref{lem:cocycle_strong} gives, for every
$A,B\in\mathcal M$ and fixed $t\in\mathbb R$,
\[
\psi_n\bigl(A\sigma_t^{\psi_n}(B)\bigr)
\longrightarrow
\omega\bigl(A\sigma_t^\omega(B)\bigr).
\]
Thus the map
\[
\psi\longmapsto\psi(A\sigma_t^\psi(B))
\]
is sequentially continuous on $O_f$. Since the predual norm topology
is metrizable, it is continuous.
\end{proof}

For self-adjoint $A,B$ and $t\in\mathbb R$, define
\[
\begin{aligned}
\mathcal I_O(A,B,t)
&=\{\operatorname{Re}\psi(A\sigma_t^\psi(B)):\psi\in O_f\},&
\mathcal J_O(A,B,t)
&=\{\operatorname{Im}\psi(A\sigma_t^\psi(B)):\psi\in O_f\},\\
\mathcal K_O(A,B,t)
&=\{\operatorname{Re}\psi(\sigma_t^\psi(B)A):\psi\in O_f\},&
\mathcal L_O(A,B,t)
&=\{\operatorname{Im}\psi(\sigma_t^\psi(B)A):\psi\in O_f\}.
\end{aligned}
\]

\begin{proposition}[Parcel correlation intervals]
\label{prop:solid}
For every normal parcel $O$, $\mathcal I_O,\mathcal J_O,\mathcal K_O$
and $\mathcal L_O$ are intervals in $\mathbb R$.
\end{proposition}

\begin{proof}
The faithful part $O_f$ is convex: a convex combination of faithful
states is faithful. Hence $O_f$ is path connected. Modular continuity
therefore makes the real and imaginary parts of
$\psi(A\sigma_t^\psi(B))$ continuous real-valued functions on a connected
space, so their images are intervals. Since
\[
\psi(\sigma_t^\psi(B)A)
=
\psi(B\sigma_{-t}^\psi(A)),
\]
the same argument gives the other two intervals.
\end{proof}

The corresponding complex value sets are connected subsets of the
rectangles $\mathcal I_O\times\mathcal J_O$ and
$\mathcal K_O\times\mathcal L_O$; they need not fill those rectangles.

We now turn to reduction. Fix a faithful normal reference state
$\omega_0$ and its modular dynamics $\sigma_t^{\omega_0}$, and let
$\mathcal A_0\subseteq\mathcal M$ be a norm-dense $*$-subalgebra of
entire analytic elements for this dynamics.

\begin{definition}[KMS defect]
\label{def:kmsdefect}
For a normal parcel $O$, $A,B\in\mathcal A_0$ and $t\in\mathbb R$, define the {\em KMS defect} as 
\[
\delta_O(A,B,t)
=
\sup_{\rho\in O}
\left|
\rho(A\sigma_{t+i}^{\omega_0}(B))
-
\rho(\sigma_t^{\omega_0}(B)A)
\right|.
\]
\end{definition}

Thus $\delta_O$ measures uniformly across the parcel the failure of the
KMS boundary identity for the fixed dynamics $\sigma_t^{\omega_0}$.

\begin{theorem}[Parcel KMS Reduction]
\label{thm:kmsreduction}
Let $\{O_\alpha\}$ be a decreasing net of normal parcels and
$C_\alpha=\overline{O_\alpha}^{\,\sigma(\mathcal M_*,\mathcal M)}$.
Assume:
\begin{enumerate}
\item[\rm(i)] $\operatorname{osc}_{O_\alpha}(C)\to0$ for every
$C\in\mathcal A_0$;
\item[\rm(ii)] $\delta_{O_\alpha}(A,B,t)\to0$ for every
$A,B\in\mathcal A_0$ and $t\in\mathbb R$;
\item[\rm(iii)] $\bigcap_\alpha C_\alpha\neq\varnothing$.
\end{enumerate}
Then $\bigcap_\alpha C_\alpha=\{\omega_*\}$ for a unique normal state
$\omega_*$, and $\omega_*$ is KMS for
$(\mathcal M,\sigma_t^{\omega_0})$ at inverse temperature $\beta=1$.
\end{theorem}

\begin{proof}
If $\rho,\eta\in\bigcap_\alpha C_\alpha$, then for every
$C\in\mathcal A_0$,
\[
|\rho(C)-\eta(C)|
\leq
\operatorname{osc}_{C_\alpha}(C)
=
\operatorname{osc}_{O_\alpha}(C)
\longrightarrow0.
\]
The equality follows from continuity of evaluation and
$C_\alpha=\overline{O_\alpha}^{\,\sigma(\mathcal M_*,\mathcal M)}$.
Norm density of $\mathcal A_0$ therefore gives $\rho=\eta$, so the
non-empty intersection is $\{\omega_*\}$.

For $A,B\in\mathcal A_0$ and fixed $t$, the two evaluations entering
$\delta_{O_\alpha}(A,B,t)$ are weak$^*$ continuous. Since
$\omega_*\in C_\alpha$, passage to the closure gives
\[
\left|
\omega_*(A\sigma_{t+i}^{\omega_0}(B))
-
\omega_*(\sigma_t^{\omega_0}(B)A)
\right|
\leq\delta_{O_\alpha}(A,B,t).
\]
Letting $\alpha$ increase yields the KMS boundary identity. The
analytic-core characterization of KMS states
\cite[Prop.~5.3.7]{brattelirobinson2} then gives the result.
\end{proof}

\begin{corollary}[Reduction to a prescribed KMS state]
\label{thm:kms_prescribed}
Let $\omega_*$ be a normal KMS state for
$\sigma_t^{\omega_0}$ at $\beta=1$, and let $\{O_\alpha\}$ be a decreasing net of normal parcels containing
$\omega_*$ such that
$\operatorname{osc}_{O_\alpha}(C)\to0$ on a norm-dense $*$-subalgebra.
Then $\bigcap_\alpha O_\alpha=\{\omega_*\}$.
\end{corollary}

\begin{proof}
This is Theorem~\ref{thm:prescribedstate}; the KMS property is already
part of the hypothesis.
\end{proof}

The two reduction results serve different purposes. When the KMS state is
known in advance, only observational shrinkage is required to recover it.
When it is not known, the KMS defect supplies the additional operational
condition forcing the limiting state to satisfy the KMS boundary
identity. The next section applies the prescribed-state form to the
Bisognano--Wichmann/Unruh setting.
\section{The Parcel Unruh Effect}
\label{sec:unruh}

The Bisognano--Wichmann theorem~\cite{bisognanowichmann76} identifies the
restriction of the Minkowski vacuum to a Rindler wedge as a KMS state for
Lorentz-boost dynamics. We now express this result at finite precision:
a parcel need not consist of thermal states, but refinement around the
vacuum recovers the exact KMS relation and Unruh temperature.

Let
\[
W_R=\{x\in\mathbb R^{1+3}:x^1>|x^0|\}
\]
be the right Rindler wedge. An observer with proper acceleration $a>0$
follows
\[
x^\mu(\tau)
=
\left(
a^{-1}\sinh(a\tau),\,
a^{-1}\cosh(a\tau),\,0,\,0
\right),
\]
where the boost parameter is $s=a\tau$. If $K$ generates the boosts, write
$\alpha_s^{W_R}(C)=e^{iKs}Ce^{-iKs}$. The Bisognano--Wichmann theorem gives
\[
\sigma_t^{\omega_0}
=
\alpha_{2\pi t}^{W_R},
\]
so the restricted vacuum
$\omega_0^W:=\omega_0|_{\mathcal M(W_R)}$ is KMS for boost dynamics at
inverse temperature $2\pi$. In proper time this corresponds to the Unruh
temperature $T_{\rm U}=a/(2\pi)$.

A \emph{Rindler parcel} is a normal parcel
$O\subseteq\mathcal S^{\sf n}(\mathcal M(W_R))$. Let
$\mathcal A_0\subseteq\mathcal M(W_R)$ be a norm-dense $*$-subalgebra of
entire analytic elements for the boost dynamics.

\begin{definition}[Thermal value sets and boost KMS defect]
\label{def:boostkmsdefect}
For a Rindler parcel $O$, $A,B\in\mathcal A_0$ and $\tau\in\mathbb R$,
define
\[
\begin{aligned}
\mathcal C_O^a(A,B,\tau)
&=
\{\rho(A\alpha^{W_R}_{a\tau+2\pi i}(B)):\rho\in O\},\\
\mathcal D_O^a(A,B,\tau)
&=
\{\rho(\alpha^{W_R}_{a\tau}(B)A):\rho\in O\},
\end{aligned}
\]
and
\[
\delta_O^a(A,B,\tau)
=
\sup_{\rho\in O}
\left|
\rho(A\alpha^{W_R}_{a\tau+2\pi i}(B))
-
\rho(\alpha^{W_R}_{a\tau}(B)A)
\right|.
\]
\end{definition}

Because the boost dynamics is fixed independently of $\rho$, these are
images of the parcel under affine weak$^*$-continuous evaluation maps.

\begin{proposition}[Vacuum-centred defect bound]
\label{prop:oscillationboundsboost}
If $\omega_0^W\in O$, then
\[
\delta_O^a(A,B,\tau)
\leq
\operatorname{osc}_O
\bigl(A\alpha^{W_R}_{a\tau+2\pi i}(B)\bigr)
+
\operatorname{osc}_O
\bigl(\alpha^{W_R}_{a\tau}(B)A\bigr).
\]
\end{proposition}

\begin{proof}
Set
$X=A\alpha^{W_R}_{a\tau+2\pi i}(B)$ and
$Y=\alpha^{W_R}_{a\tau}(B)A$. Bisognano--Wichmann gives
$\omega_0^W(X)=\omega_0^W(Y)$. Hence, for $\rho\in O$,
\[
|\rho(X)-\rho(Y)|
\leq
|\rho(X)-\omega_0^W(X)|
+
|\rho(Y)-\omega_0^W(Y)|
\leq
\operatorname{osc}_O(X)+\operatorname{osc}_O(Y),
\]
and taking the supremum gives the result.
\end{proof}

\begin{proposition}[Finite-precision Unruh correlations]
\label{thm:intervalunruh}
If $\omega_0^W\in O$, then
$\mathcal C_O^a(A,B,\tau)$ and
$\mathcal D_O^a(A,B,\tau)$ are convex subsets of $\mathbb C$ whose real
and imaginary projections are intervals. Both contain the common vacuum
KMS value
\[
\omega_0^W(A\alpha^{W_R}_{a\tau+2\pi i}(B))
=
\omega_0^W(\alpha^{W_R}_{a\tau}(B)A).
\]
\end{proposition}

\begin{proof}
Both defining maps are affine and weak$^*$ continuous, so their images of
the convex parcel $O$ are convex. The common value follows from
$\omega_0^W\in O$ and the Bisognano--Wichmann KMS identity.
\end{proof}

\begin{theorem}[Parcel Unruh reduction]
\label{thm:refinementunruh}
Let $\{O_\alpha\}$, for $\alpha\in I$,  be a decreasing net of Rindler parcels containing
$\omega_0^W$ such that
\[
\operatorname{osc}_{O_\alpha}(C)\longrightarrow0
\qquad(C\in\mathcal M(W_R)).
\]
Then
\[
\bigcap_\alpha O_\alpha=\{\omega_0^W\},
\]
and for every $A,B\in\mathcal A_0$ and $\tau\in\mathbb R$,
\[
\delta_{O_\alpha}^a(A,B,\tau)\longrightarrow0.
\]
Consequently the thermal value sets collapse to the common
Bisognano--Wichmann KMS value, and the limiting state has Unruh
temperature
\[
T_{\rm U}=\frac{a}{2\pi}.
\]
\end{theorem}
\begin{proof}
The first assertion follows from
Theorem~\ref{thm:prescribedstate}, with prescribed state $\omega_0^W$.
The defect convergence follows from
Proposition~\ref{prop:oscillationboundsboost}, since both
\[
X=A\alpha^{W_R}_{a\tau+2\pi i}(B),
\qquad
Y=\alpha^{W_R}_{a\tau}(B)A
\]
belong to $\mathcal M(W_R)$ and their oscillations therefore tend to
zero.

Moreover, since $\omega_0^W\in O_\alpha$ for every $\alpha$,
\[
\sup_{\rho\in O_\alpha}
|\rho(X)-\omega_0^W(X)|
\leq
\operatorname{osc}_{O_\alpha}(X)
\longrightarrow0,
\]
and similarly
\[
\sup_{\rho\in O_\alpha}
|\rho(Y)-\omega_0^W(Y)|
\leq
\operatorname{osc}_{O_\alpha}(Y)
\longrightarrow0.
\]
Hence $\mathcal C_{O_\alpha}^a(A,B,\tau)$ and
$\mathcal D_{O_\alpha}^a(A,B,\tau)$ both collapse to their common
vacuum KMS value. The limiting KMS identity and temperature then follow
from the Bisognano--Wichmann theorem.
\end{proof}
The finite-precision formulation clarifies the status of exact Unruh
thermality. The Bisognano--Wichmann theorem assigns an exact KMS relation
and temperature to the exactly specified vacuum restricted to the wedge,
whereas finite observations determine only a parcel of compatible states.
At finite precision the corresponding thermal correlations therefore form
value sets rather than single exact values; as the parcel is refined
around the vacuum, these sets collapse to the Bisognano--Wichmann KMS
value and the Unruh temperature $T_{\rm U}=a/(2\pi)$. Thus exact Unruh
thermality is retained as the limiting AQFT statement while IAQFT makes
explicit what remains operationally determined before that limit is
reached.
\section{Lattice Approximations in IAQFT}
\label{sec:lattice}

We consider lattice approximations that form a compatible directed system
of $C^*$-algebras whose inductive limit is the continuum observable
algebra. Under this hypothesis, finite-precision lattice data determine
continuum parcels canonically, and parcel reduction is compatible with
the continuum limit.

Let $I$ consist of pairs
\[
i=(\Lambda,a),
\]
where $\Lambda$ is a finite lattice region and $a>0$ is the lattice
spacing, ordered by
\[
(\Lambda,a)\preceq(\Lambda',a')
\quad\Longleftrightarrow\quad
\Lambda\subseteq\Lambda'
\ \text{and}\ 
a'\leq a.
\]
For $i=(\Lambda,a)\in I$, write
\[
\mathcal A_i:=\mathcal A^{(a)}(\Lambda).
\]
Suppose that $I$ is directed and that the lattice algebras are connected
by compatible injective unital $*$-homomorphisms
\[
\iota_{ij}:\mathcal A_i\longrightarrow\mathcal A_j
\qquad(i,j\in I,\ i\preceq j),
\]
with continuum quasi-local algebra
\[
\mathcal A=\varinjlim_{i\in I}\mathcal A_i.
\]
For each $i\in I$, write
$\iota_i:\mathcal A_i\to\mathcal A$ for the canonical embedding and
\[
r_i:\mathcal S(\mathcal A)\longrightarrow\mathcal S(\mathcal A_i),
\qquad
r_i(\omega)=\omega\circ\iota_i,
\]
for the corresponding restriction map.

For $i\in I$, a \emph{lattice parcel} is a non-empty convex weak$^*$ open
subset $O_i\subseteq\mathcal S(\mathcal A_i)$. A \emph{lattice
observational parcel} is one defined by finitely many finite-precision
lattice expectation constraints, exactly as in
Definition~\ref{def:quantum-parcel}.

\begin{theorem}[Lattice parcels lift to continuum parcels]
\label{thm:latticeparcel}
For every $i\in I$ and every lattice parcel
$O_i\subseteq\mathcal S(\mathcal A_i)$, its inverse image
\[
\widetilde O_i:=r_i^{-1}(O_i)
\]
is a parcel in $\mathcal S(\mathcal A)$.
\end{theorem}

\begin{proof}
The map $r_i$ is affine and weak$^*$ continuous, so $\widetilde O_i$ is
open and convex. It is non-empty because every state on the unital
$C^*$-subalgebra $\iota_i(\mathcal A_i)\subseteq\mathcal A$ extends to a
state on $\mathcal A$.
\end{proof}

\begin{example}[A lattice expectation constraint and its continuum parcel]
\label{ex:latticeparcel}
Fix $i\in I$ and let $H_i=H_i^*\in\mathcal A_i$ be a lattice observable.
Suppose a finite-precision lattice measurement determines
\[
a<\rho_i(H_i)<b.
\]
The corresponding lattice observational parcel is
\[
O_i
=
\{\rho_i\in\mathcal S(\mathcal A_i):
a<\rho_i(H_i)<b\}.
\]
Its induced continuum parcel is
\[
\widetilde O_i
=
r_i^{-1}(O_i)
=
\{\omega\in\mathcal S(\mathcal A):
a<\omega(\iota_i(H_i))<b\}.
\]
Thus a finite-resolution lattice observation determines exactly the
continuum states whose restriction to the lattice reproduces the observed
expectation interval.

If $i,j\in I$ with $i\preceq j$ and the same observable is represented
on the refined lattice by $H_j=\iota_{ij}(H_i)$, then, writing
$r_{ij}:\mathcal S(\mathcal A_j)\to\mathcal S(\mathcal A_i)$ for
restriction along $\iota_{ij}$, the corresponding refined lattice observational parcel is
\[
O_j
=
\{\rho_j\in\mathcal S(\mathcal A_j):
a<\rho_j(H_j)<b\}
=
r_{ij}^{-1}(O_i).
\]
It induces the same continuum parcel:
\[
r_j^{-1}(O_j)=r_i^{-1}(O_i).
\]
Thus refining the lattice representation of a fixed finite-precision
observation does not by itself change the continuum observational
information.
\end{example}

Thus a finite-precision lattice experiment specifies precisely the
continuum states whose restrictions to the lattice algebra satisfy the
observed constraints. The lattice spacing and the parcel widths represent
distinct approximations: the former controls the algebraic approximation
to the continuum, while the latter controls observational precision.

The construction is consistent under lattice refinement. For
$i,j\in I$ with $i\preceq j$, let
\[
r_{ij}:\mathcal S(\mathcal A_j)\longrightarrow
\mathcal S(\mathcal A_i)
\]
be restriction along $\iota_{ij}$. Compatibility of the embeddings gives
\[
r_i=r_{ij}\circ r_j.
\]

\begin{proposition}[Compatibility with lattice refinement]
\label{thm:latticecompatibility}
Let $i,j\in I$ with $i\preceq j$. If
$O_i\subseteq\mathcal S(\mathcal A_i)$ is a lattice parcel and
\[
O_j=r_{ij}^{-1}(O_i),
\]
then
\[
r_j^{-1}(O_j)=r_i^{-1}(O_i).
\]
Thus refinement of the lattice representation does not change the
continuum parcel represented by the same observational constraints.
\end{proposition}

\begin{proof}
This is immediate from $r_i=r_{ij}\circ r_j$.
\end{proof}

Finally, parcel reduction is compatible with the directed continuum
approximation.

\begin{theorem}[Scaling limit of parcel reduction]
\label{thm:scalingreduction}
For each $i\in I$, let
$O_i\subseteq\mathcal S(\mathcal A_i)$ be a lattice parcel and let
\[
\widetilde O_i:=r_i^{-1}(O_i)
\subseteq\mathcal S(\mathcal A)
\]
be its induced continuum parcel. Suppose that the net
$\{\widetilde O_i\}_{i\in I}$ is decreasing with respect to
$\preceq$, that is,
\[
i\preceq j
\quad\Longrightarrow\quad
\widetilde O_j\subseteq\widetilde O_i,
\]
and that
\[
\operatorname{osc}_{\widetilde O_i}(A)\longrightarrow0
\qquad(A\in\mathcal A)
\]
along the directed set $I$. Then
\[
\bigcap_{i\in I}
\overline{\widetilde O_i}^{\,w^*}
=
\{\omega_0\}
\]
for a unique continuum state
$\omega_0\in\mathcal S(\mathcal A)$.
\end{theorem}

\begin{proof}
This is Theorem~\ref{thm:reduction}, applied to the decreasing net
$\{\widetilde O_i\}_{i\in I}$ of induced continuum parcels.
\end{proof}

Hence, for lattice regularisations satisfying the directed-system
hypotheses above, passing to the continuum is compatible with the
finite-precision parcel formalism. The continuum and observational limits
remain conceptually distinct: lattice refinement removes the
regularisation, whereas parcel refinement removes observational
uncertainty.
\section{Minimality, Finiteness and Factor Type}
\label{sec:minimality-finiteness}

We now show that the Murray--von Neumann structure of a factor can be
read from the geometry of limiting parcel faces and their canonical
conjugation actions. Minimality is detected by the geometry of the face,
finiteness by its invariant states, and in Type~$\mathrm{II}_1$ factors
the trace can be recovered by homogeneous parcel counting. Proofs of the
results in this section are collected in
Appendix~\ref{app:parcelproofs}.

Recall that two projections $p,q\in\mathcal M$ are
\emph{Murray--von Neumann equivalent}, written $p\sim q$, if there
exists a partial isometry $v\in\mathcal M$ such that
\[
v^*v=p,
\qquad
vv^*=q.
\]

Thus $p$ and $q$ are respectively the initial and final projections of
the same partial isometry. Recall, in addition, that a nonzero projection is \emph{minimal}
if it has no nontrivial subprojection and \emph{finite} if no proper
subprojection is equivalent to it. A factor is of type~$\mathrm I$ if it
contains a minimal projection, type~$\mathrm{II}$ if it contains no
minimal projection but has a nonzero finite projection, and
type~$\mathrm{III}$ if it has no nonzero finite projection
\cite[Ch.~6]{KadisonRingrose2}.

For every nonzero projection $p\in\mathcal M$, define
\[
O_n(p)
=
\{\omega\in\mathcal S^{\sf n}(\mathcal M):
1-\tfrac1n<\omega(p)\},
\qquad
F_p
=
\bigcap_{n\geq1}O_n(p)
=
\{\omega:\omega(p)=1\}.
\]
Each $O_n(p)$ is an observational parcel, determined by a
finite-precision expectation constraint on $p$. Thus $F_p$ is the
limiting face of increasingly precise observational parcels
concentrating on $p$.

\begin{proposition}[Projection faces]
\label{prop:face}
For every nonzero projection $p$, $F_p$ is the exposed face of
$\mathcal S^{\sf n}(\mathcal M)$ defined by $\omega(p)=1$, and
\[
q\leq p
\quad\Longleftrightarrow\quad
F_q\subseteq F_p,
\]
with proper inclusion when $q<p$.
\end{proposition}

\begin{corollary}[Minimality]
\label{cor:minimal}
A nonzero projection $p$ is minimal if and only if $F_p$ is a singleton.
\end{corollary}

\begin{theorem}[Type~$\mathrm I$]
\label{thm:typeI}
A factor $\mathcal M$ is of type~$\mathrm I$ if and only if $F_p$ has an
extreme point for some nonzero projection $p$. In that case every nonzero
$F_q$ has an extreme point.
\end{theorem}

Minimality is therefore detected by the convex geometry of the limiting
face alone. Finiteness requires, in addition, the canonical internal
symmetry of the corner.

The relation between unitary conjugation and tracial states is classical.
In particular, the Dixmier property and its variants study averaging
under unitary conjugation and its relation to tracial structure; see, for
example, Archbold, Robert and Tikuisis~\cite{ArchboldRobertTikuisis2017}.
Our use of unitary conjugation is different. Rather than averaging unitary
orbits in the algebra, we use the canonical action of the unitary group of
the corner $p\mathcal Mp$ on the limiting projection face $F_p$ and study
its common fixed points.

Let $\mathcal U(p\mathcal Mp)$ denote the unitary group of the corner
$p\mathcal Mp$, whose identity is $p$. Thus, for
$U\in\mathcal U(p\mathcal Mp)$,
\[
U^*U=UU^*=p.
\]
For such a $U$, set
\[
\widehat U=U+(1-p).
\]
Then $\widehat U$ is a unitary in $\mathcal M$ extending $U$, and
\[
\widehat U p\widehat U^*=p.
\]
It therefore induces an action on $F_p$ by
\[
\omega^{\widehat U}(A)
=
\omega(\widehat U A\widehat U^*),
\qquad
\omega\in F_p,\quad A\in\mathcal M.
\]
For a fixed $U\in\mathcal U(p\mathcal Mp)$, define
\[
\operatorname{Fix}(\widehat U)
=
\{\omega\in F_p:\omega^{\widehat U}=\omega\}.
\]
This is a convex weakly closed subset of $F_p$. The states invariant
under the action of the whole unitary group are precisely those in the
common fixed-point set
\[
\operatorname{Fix}\bigl(\mathcal U(p\mathcal Mp)\bigr)
:=
\bigcap_{U\in\mathcal U(p\mathcal Mp)}
\operatorname{Fix}(\widehat U).
\]
The following theorem identifies the existence of such a common fixed
point with finiteness of the projection.

\begin{theorem}[Finiteness as a fixed point]
\label{thm:fixedpoint}
Let $\mathcal M$ be a factor and $p\neq0$. Then
\[
p\text{ is finite}
\quad\Longleftrightarrow\quad
\operatorname{Fix}\bigl(\mathcal U(p\mathcal Mp)\bigr)\neq\varnothing.
\]
When $p$ is finite,
\[
\operatorname{Fix}\bigl(\mathcal U(p\mathcal Mp)\bigr)
=
\{\tau_p\},
\]
where $\tau_p$ is the normalized normal trace on the finite factor
$p\mathcal Mp$, regarded as a state in $F_p$.
\end{theorem}

Combining this with Corollary~\ref{cor:minimal} gives
\[
p\text{ minimal}
\iff
F_p\text{ is a singleton},
\qquad
p\text{ finite}
\iff
F_p\text{ has a conjugation-invariant state}.
\]

These two criteria recover the complete factor classification. 

\begin{theorem}[Parcel classification of factor type]
\label{thm:classification}
For a factor $\mathcal M$, define the following properties of its
projection faces $F_p$:
\begin{itemize}
\item \textbf{Atoms:} some nonzero $F_p$ has an extreme point;
\item \textbf{Whole-finite:} the canonical conjugation action on
$F_1=\mathcal S^{\sf n}(\mathcal M)$ has a fixed point;
\item \textbf{Some-finite:} the canonical conjugation action on $F_p$
has a fixed point for some nonzero projection $p$.
\end{itemize}
Then the Murray--von Neumann type of $\mathcal M$ is determined by
\[
\begin{array}{lccc}
 & \textbf{Atoms} & \textbf{Whole-finite} & \textbf{Some-finite}\\
\hline
\mathrm I_n        & \mathrm T & \mathrm T & \mathrm T\\
\mathrm I_\infty   & \mathrm T & \mathrm F & \mathrm T\\
\mathrm{II}_1      & \mathrm F & \mathrm T & \mathrm T\\
\mathrm{II}_\infty & \mathrm F & \mathrm F & \mathrm T\\
\mathrm{III}       & \mathrm F & \mathrm F & \mathrm F .
\end{array}
\]
Thus the factor type is completely determined by the geometry and
canonical conjugation dynamics of the limiting projection faces.
\end{theorem}
For the Type~$\mathrm{III}$ local algebras characteristic of AQFT, this
classification has a direct finite-precision significance. Such factors
have no minimal projections and hence no pure normal states;
correspondingly, none of the nonzero projection faces $F_p$ is a
singleton. Thus the finite-dimensional picture in which sufficiently
sharp projection information can isolate a rank-one state has no local
Type~$\mathrm{III}$ analogue. Parcels require neither minimal projections
nor density-matrix representations and therefore provide finite
observational descriptions directly on the normal state space where this
familiar picture is unavailable.

The qualitative criteria above determine factor type but do not compare
finite projections within a Type~$\mathrm{II}$ factor. In a
Type~$\mathrm{II}_1$ factor this additional information is encoded by
homogeneous parcel counting.

Let $\tau$ be the normalized trace. For every $n\geq1$ choose a
decomposition
\[
1=e_1^{(n)}+\cdots+e_n^{(n)}
\]
into mutually orthogonal equivalent projections, each of trace $1/n$.
For a projection $p$, let $k_n(p)$ be the largest $k\leq n$ for which
there are mutually orthogonal projections $q_1,\ldots,q_k\leq p$ with
$q_j\sim e_1^{(n)}$.

\begin{theorem}[Parcel counting reconstructs the trace]
\label{thm:convergence}
For every projection $p$ in a Type~$\mathrm{II}_1$ factor,
\[
k_n(p)=\lfloor n\tau(p)\rfloor,
\qquad
\frac{k_n(p)}{n}\longrightarrow\tau(p).
\]
The limit is independent of the homogeneous partitions chosen.
\end{theorem}

\begin{proof}
Each $q_j$ has trace $1/n$, so $k$ mutually orthogonal such projections
below $p$ require $k/n\leq\tau(p)$ and hence
$k_n(p)\leq\lfloor n\tau(p)\rfloor$. Conversely, if
$k=\lfloor n\tau(p)\rfloor$, the continuous dimension theory of a
Type~$\mathrm{II}_1$ factor gives a projection of trace $k/n$
subequivalent to $p$, which can be decomposed into $k$ mutually
orthogonal projections of trace $1/n$. Hence
$k_n(p)=\lfloor n\tau(p)\rfloor$, and the convergence follows.
\end{proof}
\begin{example}[Recovering trace by parcel counting]
If $\tau(p)=0.37$ in a Type~$\mathrm{II}_1$ factor, then
Theorem~\ref{thm:convergence} gives
$k_n(p)=\lfloor0.37n\rfloor$. Thus, for example,
$k_{10}(p)=3$ and $k_{100}(p)=37$, while
$k_n(p)/n\to0.37$. Thus the value $\tau(p)$ is recovered from the asymptotic homogeneous
counting ratios.
\end{example}
Geometrically, $k_n(p)$ counts mutually disjoint transported copies of
the standard face $F_{e_1^{(n)}}$ inside $F_p$: if
$q_j\sim e_1^{(n)}$ and $q_j\leq p$, then $F_{q_j}\subseteq F_p$, while
a partial isometry implementing the equivalence induces the corresponding
affine homeomorphism of faces. The transported faces, rather than the
fixed $F_{e_i^{(n)}}$ themselves, must be counted because
subequivalence need not imply $e_i^{(n)}\leq p$.

\begin{corollary}[Projection equivalence from parcel counting]
\label{cor:equivcounting}
For projections $p,q$ in a Type~$\mathrm{II}_1$ factor,
\[
p\sim q
\quad\Longleftrightarrow\quad
\lim_{n\to\infty}\frac{k_n(p)}n
=
\lim_{n\to\infty}\frac{k_n(q)}n.
\]
\end{corollary}

Thus the parcel description exhibits a three-level hierarchy:
minimality is encoded by static convex geometry, finiteness by the
canonical conjugation action, and comparison of finite projections by
the continuous geometry of homogeneous parcel partitions.

\section{Open Problems and Research Directions}
\label{sec:open}

Several directions remain for further development of IAQFT. First, the
geometry of parcels on Type~$\mathrm{III}$ state spaces is largely
unexplored. Since local AQFT algebras are typically Type~$\mathrm{III}$,
it would be useful to identify intrinsic geometric invariants of parcel
refinement that require neither traces nor density-matrix
representations. The operational volume and relative information gain
introduced in Section~\ref{sec:operational-volume} provide
finite-dimensional measures associated with chosen observational
coordinates. An important open problem is to develop intrinsic
entropy-like or complexity measures directly from Type~$\mathrm{III}$
parcel geometry and to determine their relation to physical entropy and
entanglement.

A related problem is to develop parcel formulations of the finer
Connes--Takesaki invariants of Type~$\mathrm{III}$ factors. The Connes
invariant $S(\mathcal M)$ distinguishes the
Type~$\mathrm{III}_0$, Type~$\mathrm{III}_\lambda$ and
Type~$\mathrm{III}_1$ cases~\cite{Connes1973}. For a normal parcel $O$,
with faithful part
\[
O_f=O\cap\mathcal S_f(\mathcal M),
\]
one may consider the common modular spectrum
\[
S_O(\mathcal M)
=
\bigcap_{\psi\in O_f}\operatorname{Sp}(\Delta_\psi),
\]
which gives an outer approximation to $S(\mathcal M)$. It would be
interesting to determine conditions under which suitable families of
parcel spectra recover the Connes invariant.

A complementary approach is suggested by the Connes--Takesaki
continuous decomposition and flow of weights
\cite{ConnesTakesaki1977}. Passing to the crossed-product core replaces
the Type~$\mathrm{III}$ algebra by a semifinite algebra equipped with a
canonical trace-scaling dual action. This raises the question whether
parcels can be transported to the core so that its semifinite trace
geometry and dual dynamics yield finite-precision invariants recovering
the flow of weights. Such a construction would connect the parcel
geometry of Type~$\mathrm{II}$ algebras developed above with the modular
parcel dynamics of Sections~\ref{sec:modular}--\ref{sec:kms}, and could
provide a finer parcel-theoretic classification of Type~$\mathrm{III}$
factors.

Parcel equivalence also admits natural extensions. The formulation in
Section~\ref{sec:haag} concerns vacuum refinement families; analogous
notions could be developed for KMS states, charged states and other
superselection sectors, and compared with the corresponding notions of
representation equivalence. Another question is the relation between the
modular parcel duality of Proposition~\ref{prop:Jparcel} and operational
locality, particularly under Haag duality.

A more fundamental problem is to derive parcel reduction from physical
dynamics. The reduction theorems above give conditions under which
decreasing observational oscillations determine a unique limiting state,
while the finite-dimensional theory provides sufficient volume-contraction
conditions. It remains to identify measurement processes for which such
shrinking follows intrinsically from repeated observation. Informational
completeness, ergodicity, mixing and related dynamical mechanisms provide
natural candidates.

Ultimately one would like an infinite-dimensional dynamical reduction
theory requiring no finite-dimensional volume: repeated local
measurements would themselves force decreasing observational oscillation
on a state-separating family of observables and thereby determine a
unique limiting state. Such a result would connect the abstract
reduction theory of IAQFT directly to physically specified measurement
dynamics.
\section{Conclusion: Conceptual Consequences of IAQFT}
\label{sec:conclude}

We have developed Interval Algebraic Quantum Field Theory (IAQFT) from
the distinction between exact mathematical states and finite
observational information. Exact states remain indispensable objects of
AQFT, but an experiment supplies only finitely many observations at
finite resolution and therefore does not operationally specify a unique
point state. Finite experimental specifications determine observational
parcels through finitely many finite-precision expectation constraints,
while quantum parcels, or simply parcels, are non-empty convex weak$^*$
open regions of state space representing finite-precision state
information. IAQFT therefore takes parcels, rather than exactly specified
point states, as the natural operational description of experimentally
available state information. The underlying local and
operator-algebraic structure of AQFT is unchanged, and exact states are
recovered under appropriate conditions as ideal limits of successive
parcel refinement.

The general framework is compatible with measurement update, the
information order, locality and no-signalling. Global and local reduction
theorems give conditions under which parcel refinement determines a unique
limiting state. Although the AQFT state space is generally
infinite-dimensional, finitely many observations determine a
finite-dimensional observational shadow. Its volume gives an operational
measure of parcel size, while volume ratios give a relative information
gain invariant under invertible changes of coordinates within a fixed
observational space. In finite dimensions, the exact Jacobian of an
invertible L\"uders update yields a general volume formula and contraction
criterion.

The results also show that finite precision interacts with characteristic
AQFT phenomena in several distinct ways. Nonzero spacelike vacuum
correlations and strict Bell violations, when present, persist throughout
sufficiently small observational parcels and are therefore robust under
finite observational uncertainty, while nonselective local measurement
preserves spacelike-separated expectation values. Reeh--Schlieder cyclicity acquires
the operational interpretation of local reachability of finite-precision
parcels around pure normal target states. Haag-type obstructions, by
contrast, are not removed by the strong notion of finite-precision
equivalence considered here: equivalence of complete vacuum refinement
structures already forces equivalence of the limiting vacuum GNS
representations. For the Type~$\mathrm{III}$ local algebras characteristic
of AQFT, the absence of minimal projections and pure normal states is
reflected in the absence of singleton nonzero projection faces, while
parcels continue to provide finite observational descriptions without
requiring density-matrix or minimal-projection representations.

The modular structure of AQFT likewise admits a finite-precision
formulation. Modular automorphisms transport parcels, measurement update
is modularly covariant, and modular conjugation induces a duality between
normal parcels of a von Neumann algebra and its commutant. State-dependent
modular correlations are continuous under predual-norm variation of
faithful normal states, while shrinking parcels with vanishing KMS defect
recover a KMS state for fixed dynamics. In the Rindler-wedge setting,
finite-precision thermal correlations form value sets whose refinement
around the vacuum recovers the exact Bisognano--Wichmann KMS relation and
the Unruh temperature.

Two further results show how the parcel formalism interacts with the
mathematical structure surrounding AQFT. For compatible directed lattice
approximations, lattice and observational refinement remain distinct:
the former concerns approximation of the continuum observable algebra,
whereas the latter concerns increasing observational precision; under the
stated inductive-limit hypotheses the two are compatible. At the
von Neumann-algebraic level, limiting projection faces and their canonical
conjugation dynamics recover substantial Murray--von Neumann structure:
singleton faces characterize minimality, invariant states characterize
finiteness, and these criteria determine factor type. In
Type~$\mathrm{II}_1$ factors, homogeneous parcel counting reconstructs
the normalized trace and projection equivalence.

Taken together, these results show that the distinction between exact
states and finite observational information has substantive consequences
throughout AQFT. Some exact-theory properties are already robust on
finite-precision regions of state space; some acquire a natural
operational interpretation in terms of parcel reachability or
finite-precision value sets; structural obstructions such as Haag-type
inequivalence survive passage to parcels; and other exact relations
emerge under successive refinement. The resulting picture is one in
which parcels describe the state information physically available at
finite experimental resolution, while exact AQFT states provide the
indispensable mathematical idealisation recovered in the limiting
theory. IAQFT thereby places finite observational information and exact
operator-algebraic structure within a single framework, rather than
identifying experimentally available state information with an exactly
specified point of state space.
\appendix

\section{Proofs for Volume Contraction}\label{sec:proofs-vo}
This appendix collects the proofs omitted from Section~\ref{sec:finite}.

\noindent{\bf Lemma}~{\bf \ref{lem:proj}}~(Projective-linear Jacobian)
\begin{proof}
Let
\[
H:=\ker\ell,
\]
which is the tangent space of \(W\) at every point. Fix
\(x\in W\), and set
\[
p:=\ell(Lx)>0,
\qquad
y:=f(x)=\frac{Lx}{p}.
\]
Since \(\ell(y)=1\), we have \(y\in W\).

For \(\dot x\in H\), differentiation gives
\[
Df(x)[\dot x]
=
\frac{L\dot x}{p}
-
\frac{\ell(L\dot x)}{p^2}Lx
=
\frac{1}{p}
\bigl(
L\dot x-\ell(L\dot x)y
\bigr).
\]
Define
\[
P_y:\mathcal V\longrightarrow H,
\qquad
P_y(v):=v-\ell(v)y.
\]
Indeed,
\[
\ell(P_y(v))
=
\ell(v)-\ell(v)\ell(y)
=
0.
\]
Thus
\[
Df(x)
=
\frac{1}{p}\,P_yL\big|_H.
\]

It remains to compute the determinant of
\[
T:=P_yL\big|_H:H\longrightarrow H.
\]

Choose a basis \(u_1,\ldots,u_d\) of \(H\). Since \(x\in W\),
the vectors
\[
u_1,\ldots,u_d,x
\]
form a basis of \(\mathcal V\). Moreover,
\[
Tu_i
=
Lu_i-\ell(Lu_i)y.
\]
Wedging these vectors with \(y\), all terms containing two copies of
\(y\) vanish, and hence
\[
Tu_1\wedge\cdots\wedge Tu_d\wedge y
=
Lu_1\wedge\cdots\wedge Lu_d\wedge y.
\]
Since \(y=Lx/p\), the right-hand side is
\[
\frac{1}{p}
Lu_1\wedge\cdots\wedge Lu_d\wedge Lx
=
\frac{\det L}{p}\,
u_1\wedge\cdots\wedge u_d\wedge x.
\]

On the other hand, if \(\det_H T\) denotes the determinant of \(T\) on
\(H\), then
\[
Tu_1\wedge\cdots\wedge Tu_d
=
(\det_H T)\,
u_1\wedge\cdots\wedge u_d.
\]
Because \(x,y\in W\), we have \(y-x\in H\), and therefore
\[
u_1\wedge\cdots\wedge u_d\wedge y
=
u_1\wedge\cdots\wedge u_d\wedge x.
\]
Consequently,
\[
Tu_1\wedge\cdots\wedge Tu_d\wedge y
=
(\det_H T)\,
u_1\wedge\cdots\wedge u_d\wedge x.
\]
Comparing the two expressions yields
\[
\det_H T
=
\frac{\det L}{p}.
\]

Finally, since \(H\) has dimension \(d\),
\[
\det Df(x)
=
\frac{1}{p^d}\det_H T
=
\frac{\det L}{p^{d+1}}.
\]
Taking absolute values and recalling that \(p=\ell(Lx)\), we obtain
\[
\left|\det Df(x)\right|
=
\frac{|\det L|}
{\ell(Lx)^{d+1}}.
\]
\end{proof}

\noindent{\bf Lemma}~{\bf \ref{lem:detL}}~(Determinant of the linear part)
\begin{proof}
Write the polar decomposition
\[
M_j=U E_j^{1/2},
\]
where \(U\) may be chosen unitary and \(E_j^{1/2}\) is positive
semidefinite. Then
\[
L=\Ad_U\circ T_{E_j^{1/2}},
\]
where
\[
\Ad_U(\rho)=U\rho U^*,
\qquad
T_{E_j^{1/2}}(\rho)
=
E_j^{1/2}\rho E_j^{1/2}.
\]

The map \(\Ad_U\) is an isometry of the real Hilbert space
\(\Herm(\mathcal H)\) equipped with the Hilbert--Schmidt inner product
\[
\langle A,B\rangle_{\mathrm{HS}}
=
\Tr(AB).
\]
Hence its matrix in an orthonormal real basis is orthogonal, so
\[
\det_{\mathbb R}\Ad_U\in\{-1,1\}.
\]
Since the unitary group \(U(n)\) is connected and
\[
\det_{\mathbb R}\Ad_I=1,
\]
continuity implies
\[
\det_{\mathbb R}\Ad_U=1.
\]

Diagonalise \(E_j^{1/2}\):
\[
E_j^{1/2}
=
V D V^*,
\qquad
D=\diag(\sqrt{\lambda_1},\ldots,\sqrt{\lambda_n}),
\]
where \(\lambda_1,\ldots,\lambda_n\geq0\) are the eigenvalues of
\(E_j\). Since
\[
T_{E_j^{1/2}}
=
\Ad_V\circ T_D\circ\Ad_{V^*},
\]
and both conjugation maps have determinant \(1\), we have
\[
\det_{\mathbb R}T_{E_j^{1/2}}
=
\det_{\mathbb R}T_D.
\]

Let \(e_1,\ldots,e_n\) be the standard orthonormal basis and write
\[
F_{kk}=e_ke_k^*.
\]
For \(k<l\), set
\[
F_{kl}^{+}
=
\frac{e_ke_l^*+e_le_k^*}{\sqrt2},
\qquad
F_{kl}^{-}
=
\frac{i(e_ke_l^*-e_le_k^*)}{\sqrt2}.
\]
Then
\[
\{F_{kk}:1\leq k\leq n\}
\cup
\{F_{kl}^{+},F_{kl}^{-}:1\leq k<l\leq n\}
\]
is a real Hilbert--Schmidt orthonormal basis of
\(\Herm(\mathcal H)\).

In this basis,
\[
T_D(F_{kk})=\lambda_kF_{kk},
\]
and, for \(k<l\),
\[
T_D(F_{kl}^{\pm})
=
\sqrt{\lambda_k\lambda_l}\,F_{kl}^{\pm}.
\]
Thus
\[
\det_{\mathbb R}T_D
=
\prod_{k=1}^n\lambda_k
\prod_{k<l}
\left(\sqrt{\lambda_k\lambda_l}\right)^2
=
\prod_{k=1}^n\lambda_k
\prod_{k<l}\lambda_k\lambda_l.
\]
Each \(\lambda_k\) occurs once in the first product and \(n-1\) times
in the second, and therefore
\[
\det_{\mathbb R}T_D
=
\prod_{k=1}^n\lambda_k^n
=
\left(\prod_{k=1}^n\lambda_k\right)^n
=
(\det E_j)^n.
\]
Since \(\det_{\mathbb R}\Ad_U=1\), it follows that
\[
\det_{\mathbb R}L
=
(\det E_j)^n.
\]
\end{proof}

\noindent{\bf Proposition}~{\bf \ref{prop:jacobian}}~(Exact Jacobian of the L\"uders map)
\begin{proof}
Apply Lemma~\ref{lem:proj} with
\[
\mathcal V=\Herm(\mathcal H),
\qquad
\ell=\Tr,
\qquad
W=\{\rho\in\Herm(\mathcal H):\Tr\rho=1\},
\]
and
\[
L(\rho)=M_j\rho M_j^*.
\]
Since \(\Herm(\mathcal H)\) has real dimension \(n^2\), the affine
hyperplane \(W\) has dimension
\[
d=n^2-1,
\]
and hence
\[
d+1=n^2.
\]

Because \(M_j\) is invertible, the real linear map \(L\) is invertible.
Moreover, by Lemma~\ref{lem:detL},
\[
\det_{\mathbb R}L=(\det E_j)^n.
\]
Therefore Lemma~\ref{lem:proj} gives
\[
\left|\det Df_j(\rho)\right|
=
\frac{|\det_{\mathbb R}L|}
{\Tr(L(\rho))^{n^2}}
=
\frac{(\det E_j)^n}
{\Tr(M_j\rho M_j^*)^{n^2}}.
\]
Finally, cyclicity of the trace and \(E_j=M_j^*M_j\) give
\[
\Tr(M_j\rho M_j^*)
=
\Tr(M_j^*M_j\rho)
=
\Tr(E_j\rho).
\]
Hence
\[
\left|\det Df_j(\rho)\right|
=
\frac{(\det E_j)^n}
{\Tr(E_j\rho)^{n^2}}.
\]
\end{proof}

\section{Relative Modular Continuity}
\label{sec:kms-continuity-proof}
\begin{lemma}[Convergence of vector representatives and modular operators]
\label{lem:cocycle_strong}
Let $\mathcal M$ be a $\sigma$-finite von Neumann algebra in standard
form $(\mathcal M,\mathcal H,J,\mathcal P)$, and let
$\omega\in\mathcal S_f(\mathcal M)$. Let
$\{\psi_n\}\subseteq\mathcal S_f(\mathcal M)$ satisfy
\[
\|\psi_n-\omega\|_{\mathcal M_*}\longrightarrow0.
\]
Then:
\begin{enumerate}
\item[\rm(i)]
\[
\|\Omega_{\psi_n}-\Omega_\omega\|_{\mathcal H}
\longrightarrow0.
\]

\item[\rm(ii)]
For every $x\in\mathcal M$,
\[
\left\|
\Delta_{\psi_n,\omega}^{1/2}x\Omega_\omega
-
\Delta_\omega^{1/2}x\Omega_\omega
\right\|_{\mathcal H}
\longrightarrow0.
\]

\item[\rm(iii)]
\[
\Delta_{\psi_n,\omega}
\longrightarrow
\Delta_\omega
\]
in the strong resolvent sense.

\item[\rm(iv)]
For every fixed $t\in\mathbb R$,
\[
[D\psi_n:D\omega]_t
\longrightarrow I
\]
strongly, and, for every $A,B\in\mathcal M$,
\[
\psi_n\bigl(A\sigma_t^{\psi_n}(B)\bigr)
\longrightarrow
\omega\bigl(A\sigma_t^\omega(B)\bigr).
\]

\end{enumerate}
\end{lemma}

\begin{proof}
Throughout, we use the standard-form identities
\cite[Thm.~III.2.3]{brattelirobinson1}
\[
\Delta_{\psi_n,\omega}^{1/2}x\Omega_\omega
=
Jx^*\Omega_{\psi_n},
\qquad
\Delta_\omega^{1/2}x\Omega_\omega
=
Jx^*\Omega_\omega,
\qquad x\in\mathcal M.
\]

\medskip
\noindent
\emph{Part~(i).}
By the Araki--Connes--Haagerup inequality
\cite[Thm.~III.2.4]{brattelirobinson1},
\[
\|\Omega_{\psi_n}-\Omega_\omega\|_{\mathcal H}^2
\leq
\|\psi_n-\omega\|_{\mathcal M_*}.
\]
Hence
\[
\|\Omega_{\psi_n}-\Omega_\omega\|_{\mathcal H}
\longrightarrow0.
\]

\medskip
\noindent
\emph{Part~(ii).}
For $x\in\mathcal M$, the standard-form identities and the isometry of
$J$ give
\[
\begin{aligned}
&
\left\|
\Delta_{\psi_n,\omega}^{1/2}x\Omega_\omega
-
\Delta_\omega^{1/2}x\Omega_\omega
\right\|
\\
&\qquad
=
\left\|
Jx^*(\Omega_{\psi_n}-\Omega_\omega)
\right\|
\\
&\qquad
=
\left\|
x^*(\Omega_{\psi_n}-\Omega_\omega)
\right\|
\\
&\qquad
\leq
\|x\|\,
\|\Omega_{\psi_n}-\Omega_\omega\|.
\end{aligned}
\]
Part~\emph{(i)} therefore implies
\[
\Delta_{\psi_n,\omega}^{1/2}x\Omega_\omega
\longrightarrow
\Delta_\omega^{1/2}x\Omega_\omega
\]
for every $x\in\mathcal M$.

\medskip
\noindent
\emph{Part~(iii).}
The subspace $\mathcal M\Omega_\omega$ is a core for
$\Delta_\omega^{1/2}$. Indeed, the Tomita operator $S_\omega$ is the
closure of the operator initially defined on $\mathcal M\Omega_\omega$
by
\[
x\Omega_\omega\longmapsto x^*\Omega_\omega,
\qquad x\in\mathcal M,
\]
so $\mathcal M\Omega_\omega$ is a core for $S_\omega$. Since
\[
S_\omega=J\Delta_\omega^{1/2},
\]
with $J$ antiunitary and
$D(S_\omega)=D(\Delta_\omega^{1/2})$, the graph norms of
$S_\omega$ and $\Delta_\omega^{1/2}$ coincide. Hence
$\mathcal M\Omega_\omega$ is also a core for
$\Delta_\omega^{1/2}$.

Set
\[
S_n=\Delta_{\psi_n,\omega}^{1/2},
\qquad
S=\Delta_\omega^{1/2},
\qquad
D_0=\mathcal M\Omega_\omega.
\]
By part~\emph{(ii)}, for every $\xi\in D_0$,
\[
\xi\in D(S_n)
\qquad\text{and}\qquad
S_n\xi\longrightarrow S\xi.
\]
Since $D_0$ is a core for $S$ and $S$ is self-adjoint, the map
$\xi\mapsto(S-iI)\xi$ is an isometry of $(D(S),\|\cdot\|_S)$ onto
$\mathcal H$; hence the graph-norm density of $D_0$ in $D(S)$ transfers
to norm-density of $(S-iI)D_0$ in $\mathcal H$. Let
\[
\eta=(S-iI)\xi,
\qquad
\xi\in D_0.
\]
Then
\[
\begin{aligned}
(S_n-iI)^{-1}\eta-(S-iI)^{-1}\eta
&=
(S_n-iI)^{-1}(S-iI)\xi-\xi
\\
&=
(S_n-iI)^{-1}(S-S_n)\xi.
\end{aligned}
\]
Since each $S_n$ is self-adjoint,
\[
\|(S_n-iI)^{-1}\|\leq1,
\]
and therefore
\[
\left\|
(S_n-iI)^{-1}\eta-(S-iI)^{-1}\eta
\right\|
\leq
\|(S-S_n)\xi\|
\longrightarrow0.
\]
Thus
\[
(S_n-iI)^{-1}
\longrightarrow
(S-iI)^{-1}
\]
strongly on the dense subspace $(S-iI)D_0$. The resolvents are uniformly
bounded, so the convergence extends to all of $\mathcal H$. Since
convergence at one non-real point is equivalent to strong resolvent
convergence for self-adjoint operators, it follows that
\[
S_n\longrightarrow S
\]
in the strong resolvent sense.

Now fix $z\in\mathbb C\setminus\mathbb R$ and define
\[
g_z(\lambda)
=
\frac{1}{\lambda^2-z},
\qquad
\lambda\in\mathbb R.
\]
The function $g_z$ is bounded and continuous. By the functional calculus
for strong resolvent convergence
\cite[Thm.~VIII.20]{ReedSimon1980},
\[
g_z(S_n)\longrightarrow g_z(S)
\]
strongly. Since
\[
g_z(S_n)
=
(S_n^2-zI)^{-1}
=
(\Delta_{\psi_n,\omega}-zI)^{-1}
\]
and
\[
g_z(S)
=
(S^2-zI)^{-1}
=
(\Delta_\omega-zI)^{-1},
\]
we conclude that
\[
\Delta_{\psi_n,\omega}
\longrightarrow
\Delta_\omega
\]
in the strong resolvent sense.

\medskip
\noindent
\emph{Part~(iv).}
Fix $z\in\mathbb C\setminus\mathbb R$ and define
\[
f_z(\lambda)
=
\begin{cases}
\dfrac{1}{\log\lambda-z}, & \lambda>0,\\[1ex]
0, & \lambda=0.
\end{cases}
\]
Since $\operatorname{Im}z\neq0$, the denominator does not vanish for
$\lambda>0$. Moreover,
\[
f_z(\lambda)\longrightarrow0
\qquad(\lambda\downarrow0),
\qquad
f_z(\lambda)\longrightarrow0
\qquad(\lambda\to\infty).
\]
Hence $f_z\in C_0([0,\infty))$. Extending $f_z$ by zero to
$(-\infty,0]$ gives a function, still denoted $f_z$, in
$C_0(\mathbb R)$.

By part~\emph{(iii)} and the functional calculus for strong resolvent
convergence
\cite[Thm.~VIII.20]{ReedSimon1980},
\[
f_z(\Delta_{\psi_n,\omega})
\longrightarrow
f_z(\Delta_\omega)
\]
strongly. Because $\psi_n$ and $\omega$ are faithful,
$\Delta_{\psi_n,\omega}$ and $\Delta_\omega$ are injective. Hence
\[
E_{\Delta_{\psi_n,\omega}}(\{0\})
=
E_{\Delta_\omega}(\{0\})
=
0.
\]
Since these operators are positive and injective,
$\log\Delta_{\psi_n,\omega}$ and $\log\Delta_\omega$ are therefore
well-defined self-adjoint operators by the Borel functional calculus.
Therefore the chosen value $f_z(0)=0$ does not affect the functional
calculus, and
\[
f_z(\Delta_{\psi_n,\omega})
=
(\log\Delta_{\psi_n,\omega}-zI)^{-1},
\qquad
f_z(\Delta_\omega)
=
(\log\Delta_\omega-zI)^{-1}.
\]
Hence
\[
\log\Delta_{\psi_n,\omega}
\longrightarrow
\log\Delta_\omega
\]
in the strong resolvent sense. Note that $0$ may lie in
$\sigma(\Delta_{\psi_n,\omega})$ or $\sigma(\Delta_\omega)$ without
affecting this conclusion: what is required is not a spectral gap but
the vanishing of the spectral projection at $0$, which faithfulness
already guarantees.

Applying the functional calculus for strong resolvent convergence
\cite[Thm.~VIII.20]{ReedSimon1980}, for every fixed $t\in\mathbb R$,
\[
e^{it\log\Delta_{\psi_n,\omega}}
\longrightarrow
e^{it\log\Delta_\omega}
\]
strongly. Equivalently,
\[
\Delta_{\psi_n,\omega}^{it}
\longrightarrow
\Delta_\omega^{it}
\]
strongly.

Using the Connes cocycle formula
\[
[D\psi_n:D\omega]_t
=
\Delta_{\psi_n,\omega}^{it}\Delta_\omega^{-it},
\]
and the unitarity of $\Delta_\omega^{-it}$, we obtain
\[
[D\psi_n:D\omega]_t
\longrightarrow I
\]
strongly.

Finally, put
\[
u_n=[D\psi_n:D\omega]_t,
\qquad
X=\sigma_t^\omega(B).
\]
The Connes cocycle relation gives
\[
\sigma_t^{\psi_n}(B)
=
u_nXu_n^*.
\]
Since $u_n\to I$ strongly and the $u_n$ are unitary,
$u_n^*\to I$ strongly. Consequently
\[
u_nXu_n^*
\longrightarrow X
\]
strongly, while
\[
\|u_nXu_n^*\|=\|X\|=\|B\|.
\]

By part~\emph{(i)},
\[
\Omega_{\psi_n}\longrightarrow\Omega_\omega
\]
in norm. Therefore
\[
\begin{aligned}
\psi_n\bigl(A\sigma_t^{\psi_n}(B)\bigr)
&=
\left\langle
\Omega_{\psi_n},
Au_nXu_n^*\Omega_{\psi_n}
\right\rangle
\\
&\longrightarrow
\left\langle
\Omega_\omega,
AX\Omega_\omega
\right\rangle
\\
&=
\omega\bigl(A\sigma_t^\omega(B)\bigr).
\end{aligned}
\]
This proves part~\emph{(iv)} and completes the proof.
\end{proof}

\section{Proofs for the Parcel Characterisation of Factor Type}
\label{app:parcelproofs}

This appendix collects the proofs omitted from
Section~\ref{sec:minimality-finiteness}, together with the auxiliary
lemmas and facts needed for them.

\subsection{Projection Faces and Minimality}

\begin{lemma}[Identification of $F_r$]
\label{lem:limit}
$F_r=\{\omega\in\Sn(\mathcal M):\omega(r)=1\}$.
\end{lemma}

\begin{proof}
Since $0\le\omega(r)\le1$ always, $\omega\in O_n(r)$ for all $n$ iff
$\omega(r)=1$.
\end{proof}

\begin{lemma}[Support identity]
\label{lem:support}
For $\omega\in F_r$ and $A\in\mathcal M$, $\omega(A)=\omega(rAr)$.
\end{lemma}

\begin{proof}
Since $\omega(1-r)=0$, Cauchy--Schwarz gives $\omega((1-r)A)=0$ and
$\omega(A(1-r))=0$ for every $A$; expanding
$A=(r+(1-r))A(r+(1-r))$ and applying $\omega$ leaves only $\omega(rAr)$.
\end{proof}

\begin{definition}[Face]
A convex subset $F$ of a convex set $K$ is a \emph{face} if $\omega\in
F$, $\omega=t\omega_1+(1-t)\omega_2$ with $\omega_1,\omega_2\in K$,
$t\in(0,1)$ forces $\omega_1,\omega_2\in F$.
\end{definition}

\noindent{\bf Proposition}~{\bf \ref{prop:face}}~(Projection faces)
\begin{proof}
Since
$
0\le \omega(p)\le 1
$
for every normal state $\omega$, the functional
$
\omega\mapsto\omega(p)
$
has maximum value $1$, and
\[
F_p
=
\{\omega\in\Sn(\mathcal M):\omega(p)=1\}
\]
is precisely its maximizing set. Hence $F_p$ is exposed.

If
\[
\omega=t\omega_1+(1-t)\omega_2\in F_p,
\qquad
0<t<1,
\]
then
\[
1
=
t\,\omega_1(p)+(1-t)\,\omega_2(p),
\]
with
$
\omega_1(p),\omega_2(p)\le 1.
$
Therefore
\[
\omega_1(p)=\omega_2(p)=1,
\]
so
$
\omega_1,\omega_2\in F_p.
$
Thus $F_p$ is a face.

If $q\le p$ and $\omega\in F_q$, then
\[
1=\omega(q)\le\omega(p)\le1,
\]
hence
$
\omega(p)=1,
$
so
$
F_q\subseteq F_p.
$

Conversely, suppose
\[
F_q\subseteq F_p.
\]
If $q\not\le p$, then
\[
q(1-p)q\neq0.
\]
Since $q(1-p)q$ is a nonzero positive element of the corner
$q\mathcal Mq$, there exists a normal state $\varphi$ on
$q\mathcal Mq$ such that
\[
\varphi\bigl(q(1-p)q\bigr)>0.
\]
Define a normal state on $\mathcal M$ by
\[
\omega(A)=\varphi(qAq).
\]
Then
\[
\omega(q)=1,
\]
so $\omega\in F_q$. On the other hand,
\[
1-\omega(p)
=
\omega(1-p)
=
\varphi\bigl(q(1-p)q\bigr)
>0,
\]
and hence
\[
\omega(p)<1.
\]
Thus $\omega\notin F_p$, contradicting
$F_q\subseteq F_p$. Therefore
\[
q\le p.
\]
Finally, if $q<p$, then the nonzero projection
$
p-q
$
supports a normal state belonging to
\[
F_p\setminus F_q,
\]
so the inclusion is proper.
\end{proof}

\noindent{\bf Corollary}~{\bf \ref{cor:minimal}}~(Minimality)
\begin{proof}
We prove both directions.

$(\Rightarrow)$ Suppose $p$ is minimal. Then every self-adjoint element
of the corner $p\mathcal Mp$ is a real scalar multiple of $p$: indeed,
if $X=X^*\in p\mathcal Mp$, the spectral projections of $X$ lie under
$p$, and minimality of $p$ forces every spectral projection to be either
$0$ or $p$, so $X=ap$ for some $a\in\mathbb R$.

Now let $A\in\mathcal M$. Applying this to
\[
\operatorname{Re}(A)=\frac{A+A^*}{2},
\qquad
\operatorname{Im}(A)=\frac{A-A^*}{2i},
\]
gives real numbers $a,b$ such that
\[
p\,\operatorname{Re}(A)\,p=ap,
\qquad
p\,\operatorname{Im}(A)\,p=bp.
\]
Hence
\[
pAp=(a+ib)p.
\]
Thus $p\mathcal Mp=\mathbb Cp$.

Let $\omega\in F_p$. By the support identity,
\[
\omega(A)=\omega(pAp)=(a+ib)\omega(p)=a+ib.
\]
This value depends only on $A$, not on $\omega$. Therefore all states in
$F_p$ agree on every $A\in\mathcal M$, and $F_p$ is a singleton.

\medskip
\noindent
$(\Leftarrow)$ If $F_p$ is a singleton then $p$ is minimal. We prove the
contrapositive: if $p$ is not minimal, then $F_p$ contains at least two
distinct points.

If $p$ is not minimal, there exists a projection $q$ with
\[
0<q<p,
\]
i.e.\ $q\ne0$, $q\ne p$, and $q\le p$; then $p-q$ is also a nonzero
projection with $p-q\le p$.

Since $q$ and $p-q$ are nonzero projections, each supports a normal
state: choose $\omega_1,\omega_2\in\Sn(\mathcal M)$ with
\[
\omega_1(q)=1,
\qquad
\omega_2(p-q)=1.
\]

Since $q\le p$,
\[
\omega_1(p)\ge\omega_1(q)=1,
\]
hence $\omega_1(p)=1$ and $\omega_1\in F_p$. Since $p-q\le p$, similarly
\[
\omega_2(p)\ge\omega_2(p-q)=1,
\]
hence $\omega_2(p)=1$ and $\omega_2\in F_p$.

Finally $\omega_1\ne\omega_2$: indeed $\omega_1(q)=1$, while
\[
\omega_2(q)
=
\omega_2\bigl(p-(p-q)\bigr)
=
\omega_2(p)-\omega_2(p-q)
=
1-1
=
0.
\]
Hence $F_p$ contains the two distinct points $\omega_1,\omega_2$,
completing the proof.
\end{proof}

\begin{fact}[Extreme points are pure normal states]
\label{fact:extreme}
The extreme points of $\Sn(\mathcal M)$ are exactly the pure normal
states, and $\omega$ is pure iff its support projection is minimal.
(\cite[Thm.~7.3.4]{KadisonRingrose1}.)
\end{fact}

\noindent{\bf Theorem}~{\bf \ref{thm:typeI}}~(Type~$\mathrm I$)
\begin{proof}
Suppose first that $\mathcal M$ is of type $\mathrm I$. Then every
nonzero projection $p\in\mathcal M$ dominates a minimal projection
$q\le p$. Any normal state $\omega$ supported on $q$ is pure by
Fact~\ref{fact:extreme}, and
\[
\omega\in F_q\subseteq F_p
\]
by Proposition~\ref{prop:face}. Since $\omega$ is an extreme point of
$\Sn(\mathcal M)$, it is in particular an extreme point of $F_p$.
Thus every nonzero exposed face $F_p$ has an extreme point.

Conversely, suppose that $F_p$ has an extreme point $\omega$ for some
nonzero projection $p\in\mathcal M$. Since $F_p$ is a face of
$\Sn(\mathcal M)$, every extreme point of $F_p$ is also an extreme
point of $\Sn(\mathcal M)$. Hence $\omega$ is a pure normal state.
By Fact~\ref{fact:extreme}, its support projection $s(\omega)$ is
minimal. Thus $\mathcal M$ contains a minimal projection. Since
$\mathcal M$ is a factor, this implies that $\mathcal M$ is of type
$\mathrm I$.

We have therefore proved that $\mathcal M$ is of type $\mathrm I$ if
and only if $F_p$ has an extreme point for some nonzero projection
$p$. Moreover, if this condition holds, then $\mathcal M$ is of type
$\mathrm I$, and hence every nonzero projection $q$ dominates a
minimal projection. By the first part of the proof, $F_q$ therefore
has an extreme point for every nonzero $q\in\mathcal M$.
\end{proof}
\subsection{Finiteness and Invariant States}

\begin{lemma}[The lifted action preserves $F_p$]
\label{lem:lift}
For a unitary $U\in p\mathcal Mp$, $\widehat U:=U+(1-p)$ is unitary in
$\mathcal M$, and $\omega\mapsto\omega^{\widehat U}$ maps $F_p$ affinely
and weak*-continuously into $F_p$.
\end{lemma}

\begin{proof}
Since $U\in p\mathcal Mp$,
\[
U=pUp=Up=pU,
\qquad
U^*(1-p)=(1-p)U=0.
\]
Hence
\[
\widehat U^*\widehat U
=
U^*U+(1-p)
=
p+(1-p)
=
1,
\]
and symmetrically $\widehat U\widehat U^*=1$, so $\widehat U$ is unitary.

Also,
\[
\widehat Up\widehat U^*
=
(Up+(1-p)p)(U^*+1-p)
=
UU^*+U(1-p)
=
p+0
=
p.
\]
Thus, for $\omega\in F_p$,
\[
\omega^{\widehat U}(p)
=
\omega(\widehat Up\widehat U^*)
=
\omega(p)
=
1,
\]
so $\omega^{\widehat U}\in F_p$.
\end{proof}

\begin{lemma}[Fixed points are tracial elements]
\label{lem:tracial}
$\omega\in F_p$ is a fixed point of the conjugation action of
$\mathcal U(p\mathcal Mp)$ iff $\omega|_{p\mathcal Mp}$ is tracial.
\end{lemma}

\begin{proof}
If $\omega|_{p\mathcal Mp}$ is tracial, then for any unitary
$U\in p\mathcal Mp$ and $A\in p\mathcal Mp$,
\[
\omega^{\widehat U}(A)=\omega(UAU^*)=\omega(U^*UA)=\omega(A),
\]
using traciality on the pair $(UA,U^*)$ and $U^*U=p$; agreement on
$p\mathcal Mp$ extends to all of $\mathcal M$ by Lemma~\ref{lem:support}.

Conversely, suppose that $\omega$ is fixed under the conjugation action.
Then for every unitary $U\in p\mathcal Mp$ and every $A\in p\mathcal Mp$,
\[
\omega(UAU^*)=\omega(A).
\]
Replacing $A$ by $AU$ gives
\[
\omega(UA)=\omega(AU).
\]
Since every element of a unital $C^*$-algebra is a linear combination
of unitaries, it follows by linearity that
\[
\omega(AB)=\omega(BA)
\]
for all $A,B\in p\mathcal Mp$. Hence
$\omega|_{p\mathcal Mp}$ is tracial.
\end{proof}

\begin{remark}[A shorter route to the converse]
The converse direction above can be shortened to a single citation: a
normal state invariant under every inner automorphism of a von Neumann
algebra is exactly a tracial state (\cite[Prop.~8.5.16]{KadisonRingrose2}; \cite[\S1.9]{Sakai}). The proof above spells this out in full for $p\mathcal Mp$ so
that the lemma is self-contained.
\end{remark}

\begin{fact}[Finiteness via existence of a trace]
\label{fact:finitetrace}
Let $\mathcal N$ be a factor. Then $\mathcal N$ is finite if and only
if it admits a normal tracial state. When it exists, the normalised
normal tracial state is unique: the dimension function on projections
of $\mathcal N$ uniquely determines the value of any normal trace on
projections, and hence on all positive operators by the spectral
theorem, so two normalised normal tracial states must agree.
(\cite[Ch.~III, \S7, Thm.~1--2]{Dixmier}; \cite[Thm.~2.5.3]{Sakai}.)
\end{fact}

\noindent{\bf Theorem}~{\bf \ref{thm:fixedpoint}}~(Finiteness as a fixed point)
\begin{proof}
Since $\mathcal M$ is a factor and $p\ne0$, the corner
$p\mathcal Mp$ is again a factor, with identity $p$.

By Lemma~\ref{lem:tracial}, a fixed point of the conjugation action on
$F_p$ exists if and only if $p\mathcal Mp$ admits a normal tracial
state. By Fact~\ref{fact:finitetrace}, this holds if and only if
$p\mathcal Mp$ is finite.

Now $p$ is finite in $\mathcal M$ if and only if $p\mathcal Mp$ is
finite. Indeed, partial isometries whose initial and final projections
are dominated by $p$ belong to $p\mathcal Mp$, so Murray--von Neumann
equivalence of subprojections of $p$ is the same whether computed in
$\mathcal M$ or in $p\mathcal Mp$. Hence a fixed point exists in $F_p$
if and only if $p$ is finite.

Finally, when $p$ is finite, the finite factor $p\mathcal Mp$ has a
unique normalised normal tracial state. By the support identity, a
state in $F_p$ is uniquely determined by its restriction to
$p\mathcal Mp$. Therefore uniqueness of the normalised normal tracial
state on $p\mathcal Mp$ gives uniqueness of the fixed point in $F_p$.
\end{proof}

\subsection{Factor Classification}

\begin{fact}[Type $\mathrm{II}_\infty$ contains finite projections;
type $\mathrm{III}$ does not]
\label{fact:IIinfinity}
A type $\mathrm{II}_\infty$ factor contains nonzero finite projections.
In a type $\mathrm{III}$ factor every nonzero projection is properly
infinite. (\cite[Ch.~V, \S1]{Takesaki1}.)
\end{fact}

\noindent{\bf Theorem}~{\bf \ref{thm:classification}}~(Parcel classification of factor type)
\begin{proof}
By Theorem~\ref{thm:typeI}, \textbf{Atoms} holds exactly for factors of
type $\mathrm I$.

By Theorem~\ref{thm:fixedpoint} applied to $p=1$,
\textbf{Whole-finite} holds exactly when the identity projection is
finite, equivalently when $\mathcal M$ is a finite factor. Thus it
holds exactly for types $\mathrm I_n$ and $\mathrm{II}_1$.

Again by Theorem~\ref{thm:fixedpoint}, \textbf{Some-finite} holds
exactly when $\mathcal M$ contains a nonzero finite projection. This is
true for types $\mathrm I_n$ and $\mathrm{II}_1$ by taking $p=1$; for
type $\mathrm I_\infty$ because it contains minimal projections, which
are finite; and for type $\mathrm{II}_\infty$ by
Fact~\ref{fact:IIinfinity}. It fails for type $\mathrm{III}$, since a
type $\mathrm{III}$ factor has no nonzero finite projections.

The resulting truth values are therefore exactly those displayed in the
table. The five rows are pairwise distinct. Reading off the table gives
\[
\mathcal M\text{ is type }\mathrm I
\iff
\textbf{Atoms},
\]
\[
\mathcal M\text{ is type }\mathrm{III}
\iff
\neg\textbf{Atoms}\wedge\neg\textbf{Some-finite},
\]
and
\[
\mathcal M\text{ is type }\mathrm{II}
\iff
\neg\textbf{Atoms}\wedge\textbf{Some-finite}.
\]
\end{proof}

\section*{Declarations}

\subsection*{Funding}
No funding was received for this work.

\subsection*{Competing interests}
The author declares no competing interests.

\bibliographystyle{plain}
\bibliography{library_ord}

@article{ArchboldRobertTikuisis2017,
  author  = {Archbold, Robert and Robert, Leonel and Tikuisis, Aaron},
  title   = {The {Dixmier} Property and Tracial States for
             {$C^*$}-Algebras},
  journal = {Journal of Functional Analysis},
  volume  = {273},
  number  = {8},
  pages   = {2655--2718},
  year    = {2017},
  doi     = {10.1016/j.jfa.2017.06.026}
}

@article{ConnesTakesaki1977,
  author  = {Connes, Alain and Takesaki, Masamichi},
  title   = {The Flow of Weights on Factors of Type {III}},
  journal = {Tohoku Mathematical Journal},
  series  = {2},
  volume  = {29},
  number  = {4},
  pages   = {473--575},
  year    = {1977},
  doi     = {10.2748/tmj/1178240493}
}

@article{Halvorson2001,
  author  = {Halvorson, Hans},
  title   = {{Reeh--Schlieder} Defeats {Newton--Wigner}: On Alternative
             Localization Schemes in Relativistic Quantum Field Theory},
  journal = {Philosophy of Science},
  volume  = {68},
  number  = {1},
  pages   = {111--133},
  year    = {2001},
  doi     = {10.1086/392869}
}

@article{Connes1973,
  author  = {Connes, Alain},
  title   = {Une classification des facteurs de type {III}},
  journal = {Annales scientifiques de l'\'Ecole Normale Sup\'erieure},
  series  = {4},
  volume  = {6},
  number  = {2},
  pages   = {133--252},
  year    = {1973},
  doi     = {10.24033/asens.1247}
}

@article{Yngvason2005,
  author  = {Yngvason, Jakob},
  title   = {The Role of type {III} factors in Quantum Field Theory},
  journal = {Reports on Mathematical Physics},
  volume  = {55},
  number  = {1},
  pages   = {135--147},
  year    = {2005},
  doi     = {10.1016/S0034-4877(05)80009-6}
}

@article{EarmanFraser2006,
  author  = {Earman, John and Fraser, Doreen},
  title   = {Haag's Theorem and Its Implications for the Foundations
             of Quantum Field Theory},
  journal = {Erkenntnis},
  volume  = {64},
  number  = {3},
  pages   = {305--344},
  year    = {2006},
  doi     = {10.1007/s10670-005-5814-y}
}

@article{Unruh1976,
  author  = {Unruh, W. G.},
  title   = {Notes on Black-Hole Evaporation},
  journal = {Physical Review D},
  volume  = {14},
  number  = {4},
  pages   = {870--892},
  year    = {1976},
  doi     = {10.1103/PhysRevD.14.870}
}

@article{CrispinoHiguchiMatsas2008,
  author  = {Crispino, Lu{\'i}s C. B. and Higuchi, Atsushi and
             Matsas, George E. A.},
  title   = {The Unruh Effect and Its Applications},
  journal = {Reviews of Modern Physics},
  volume  = {80},
  number  = {3},
  pages   = {787--838},
  year    = {2008},
  doi     = {10.1103/RevModPhys.80.787}
}

@incollection{FewsterRejzner2020,
  author    = {Fewster, Christopher J. and Rejzner, Kasia},
  title     = {Algebraic Quantum Field Theory: An Introduction},
  booktitle = {Progress and Visions in Quantum Theory in View of Gravity:
               Bridging Foundations of Physics and Mathematics},
  editor    = {Finster, Felix and Giulini, Domenico and Kleiner, Johannes
               and Tolksdorf, J{\"u}rgen},
  pages     = {1--61},
  publisher = {Birkh{\"a}user},
  address   = {Cham},
  year      = {2020},
  doi       = {10.1007/978-3-030-38941-3_1},
  eprint    = {1904.04051},
  archivePrefix = {arXiv},
  primaryClass  = {hep-th}
}

@article{Witten2018,
  author  = {Witten, Edward},
  title   = {{APS} Medal for Exceptional Achievement in Research: Invited article on entanglement properties of quantum field theory},
  journal = {Reviews of Modern Physics},
  volume  = {90},
  number  = {4},
  pages   = {045003},
  year    = {2018},
  doi     = {10.1103/RevModPhys.90.045003}
}

@article{VALENTE2014147,
title = {Does the {Reeh–Schlieder} theorem violate relativistic causality?},
journal = {Studies in History and Philosophy of Science Part B: Studies in History and Philosophy of Modern Physics},
volume = {48},
pages = {147-155},
year = {2014},
issn = {1355-2198},
doi = {https://doi.org/10.1016/j.shpsb.2014.05.006},
url = {https://www.sciencedirect.com/science/article/pii/S1355219814000574},
author = {Giovanni Valente}
}

@article{DelSantoGisin2025,
  author  = {Del Santo, Flavio and Gisin, Nicolas},
  title   = {Which Features of Quantum Physics Are Not Fundamentally Quantum but Are Due to Indeterminism?},
  journal = {Quantum},
  volume  = {9},
  pages   = {1686},
  year    = {2025},
  doi     = {10.22331/q-2025-04-03-1686}
}

@book{Rejzner2016,
  author    = {Rejzner, Kasia},
  title     = {Perturbative Algebraic Quantum Field Theory: An Introduction for Mathematicians},
  series    = {Mathematical Physics Studies},
  publisher = {Springer},
  address   = {Cham},
  year      = {2016},
  doi       = {10.1007/978-3-319-25901-7},
  isbn      = {978-3-319-25899-7}
}

@book{KadisonRingrose1,
  author    = {Kadison, Richard V. and Ringrose, John R.},
  title     = {Fundamentals of the Theory of Operator Algebras,
               Volume {I}: Elementary Theory},
  series    = {Graduate Studies in Mathematics},
  volume    = {15},
  publisher = {American Mathematical Society},
  address   = {Providence, RI},
  year      = {1997},
  note      = {Reprint of the 1983 original}
}

@misc{EdalatErgodicity2026,
  author       = {Edalat, Abbas},
  title        = {Quantum Ergodicity and Thermalization in Interval Quantum Mechanics},
  year         = {2026},
  howpublished = {arXiv:2606.00749 [quant-ph]}
}

@article{Treppiedi2026,
  author  = {Treppiedi, Attilio},
  title   = {Finite-resolution physics from additivity symmetry},
  journal = {International Journal of Geometric Methods in Modern Physics},
  year    = {2026},
  volume  = {2026},
  pages   = {2650117},
  doi     = {10.1142/S0219887826501173}
}

@book{sontag2013mathematical,
  title={Mathematical control theory: deterministic finite dimensional systems},
  author={Sontag, Eduardo D},
  volume={6},
  year={2013},
  publisher={Springer Science \& Business Media}
}

@book{KadisonRingrose2,
  author    = {Kadison, Richard V. and Ringrose, John R.},
  title     = {Fundamentals of the Theory of Operator Algebras,
               Volume {II}: Advanced Theory},
  series    = {Graduate Studies in Mathematics},
  volume    = {16},
  publisher = {American Mathematical Society},
  address   = {Providence, RI},
  year      = {1997},
  note      = {Reprint of the 1986 original}
}

@book{Takesaki1,
  author    = {Takesaki, Masamichi},
  title     = {Theory of Operator Algebras {I}},
  series    = {Encyclopaedia of Mathematical Sciences},
  volume    = {124},
  publisher = {Springer-Verlag},
  address   = {Berlin},
  year      = {2002},
  note      = {Reprint of the 1979 original}
}

@article{HaagKastler,
  author  = {Haag, Rudolf and Kastler, Daniel},
  title   = {An Algebraic Approach to Quantum Field Theory},
  journal = {Journal of Mathematical Physics},
  volume  = {5},
  pages   = {848--861},
  year    = {1964}
}

@book{Araki,
  author    = {Araki, Huzihiro},
  title     = {Mathematical Theory of Quantum Fields},
  publisher = {Oxford University Press},
  address   = {Oxford},
  year      = {1999}
}

@book{brattelirobinson2,
  author    = {Bratteli, Ola and Robinson, Derek W.},
  title     = {Operator Algebras and Quantum Statistical Mechanics~2:
               Equilibrium States. Models in Quantum Statistical
               Mechanics},
  series    = {Texts and Monographs in Physics},
  edition   = {Second},
  publisher = {Springer-Verlag},
  address   = {Berlin},
  year      = {1997},
  isbn      = {3-540-61443-5},
  doi       = {10.1007/978-3-662-03444-6},
}

@book{brattelirobinson1,
  author    = {Bratteli, Ola and Robinson, Derek W.},
  title     = {Operator Algebras and Quantum Statistical Mechanics~1:
               {C}*- and {W}*-Algebras. Symmetry Groups.
               Decomposition of States},
  series    = {Texts and Monographs in Physics},
  edition   = {Second},
  publisher = {Springer-Verlag},
  address   = {Berlin},
  year      = {1987},
  isbn      = {0-387-17093-6},
  doi       = {10.1007/978-3-662-02520-8},
}

@article{SummersWerner,
  author  = {Summers, Stephen J. and Werner, Reinhard F.},
  title   = {The Vacuum Violates {B}ell's Inequalities},
  journal = {Physics Letters A},
  volume  = {110(5)},
  pages   = {257--259},
  year    = {1985},
doi={10.1016/0375-9601(85)90093-3}
}

@book{Sakai,
  author    = {Sakai, Sh\^{o}ichir\^{o}},
  title     = {{$C^*$}-Algebras and {$W^*$}-Algebras},
  series    = {Ergebnisse der Mathematik und ihrer Grenzgebiete},
  volume    = {60},
  publisher = {Springer-Verlag},
  address   = {Berlin},
  year      = {1971}
}

@book{Dixmier,
  author    = {Dixmier, Jacques},
  title     = {von {N}eumann Algebras},
  series    = {North-Holland Mathematical Library},
  volume    = {27},
  publisher = {North-Holland Publishing Co.},
  address   = {Amsterdam},
  year      = {1981},
  note      = {Translated from the French by F.~Jellett}
}

@book{streaterWightman,
  author    = {Streater, Raymond F. and Wightman, Arthur S.},
  title     = {{PCT}, Spin and Statistics, and All That},
  publisher = {W. A. Benjamin},
  address   = {New York},
  year      = {1964}
}

@article{bisognanowichmann76,
  author  = {Bisognano, Joseph J. and Wichmann, Eyvind H.},
  title   = {On the Duality Condition for Quantum Fields},
  journal = {Journal of Mathematical Physics},
  volume  = {17},
  number  = {3},
  pages   = {303--321},
  year    = {1976},
  doi     = {10.1063/1.522898}
}

@article{Wightman1956,
  author  = {Wightman, Arthur S.},
  title   = {Quantum Field Theory in Terms of Vacuum Expectation Values},
  journal = {Physical Review},
  volume  = {101},
  number  = {2},
  pages   = {860--866},
  year    = {1956},
  doi     = {10.1103/PhysRev.101.860}
}

@misc{Edalat2026,
  author        = {Abbas Edalat},
  title         = {Finite-Precision Quantum Mechanics: Quantum Parcels, Information and Geometry},
  year          = {2026},
  note          = {arXiv:2605.19706 [quant-ph], \url{https://arxiv.org/abs/2605.19706}}
}

@book{vonNeumann1955,
  author    = {John von Neumann},
  title     = {Mathematical Foundations of Quantum Mechanics},
  publisher = {Princeton University Press},
  year      = {1955},
  note      = {Classical rigorous treatment of density operators and quantum statistics.}
}

@book{ReedSimon1980,
  author    = {Michael Reed and Barry Simon},
  title     = {Methods of Modern Mathematical Physics, {Vol. I:} Functional Analysis},
  edition   = {2nd},
  publisher = {Academic Press},
  year      = {1980}
  }

@article{Heisenberg1925,
    author = {Heisenberg, Werner},
    title = {Quantum-theoretical re-interpretation of kinematic and mechanical relations},
    journal = {Zeitschrift für Physik},
    volume = {33},
    pages = {879--893},
    year = {1925},
    note = {English translation in: B. L. van der Waerden (ed.), \emph{Sources of Quantum Mechanics}, North-Holland, 1967, pp. 261--276}
}

@inproceedings{smyth1983,
  title={Power domains and predicate transformers: A topological view},
  author={Smyth, Michael B},
  booktitle={International Colloquium on Automata, Languages, and Programming},
  pages={662--675},
  year={1983},
  organization={Springer}
}

@article{abramsky1991domain,
  title={Domain theory in logical form},
  author={Abramsky, Samson},
  journal={Annals of pure and applied logic},
  volume={51},
  number={1-2},
  pages={1--77},
  year={1991},
  publisher={Elsevier}
}

@incollection{abramsky_domain_1994,
  author    = {Abramsky, Samson and Jung, Achim},
  title     = {Domain {Theory}},
  booktitle = {Handbook of {Logic} in {Computer} {Science}},
  volume    = {3},
  editor    = {Abramsky, Samson and Gabbay, Dov M. and Maibaum, Thomas S. E.},
  publisher = {Oxford University Press},
  year      = {1994},
  pages     = {1--168},
  url       = {https://www.cs.bham.ac.uk/~axj/pub/papers/handy1.pdf}
}

@book{johnstone1982stone,
  title={Stone spaces},
  author={Johnstone, Peter T},
  year={1982},
  publisher={Cambridge University Press}
}

@book{haag2012local,
  title={Local quantum physics: Fields, particles, algebras},
  author={Haag, Rudolf},
  year={2012},
  publisher={Springer Science \& Business Media}
}

@book{scott1970outline,
  title={Outline of a mathematical theory of computation},
  author={Scott, Dana},
  year={1970},
  publisher={Oxford University Computing Laboratory, Programming Research Group Oxford, UK}
}

@book{vic89,
  author =	 {Vickers, Steve},
  publisher =	 {Cambridge University Press},
  title =	 {{Topology via Logic}},
  year=1989
}

\end{document}